%% file: tc_sub.tex
\documentclass[11pt]{article}
\usepackage{amsmath,amssymb}
\usepackage{amsthm}
\usepackage{mathtools}
\usepackage[noend]{algorithmic}
\usepackage[ruled,vlined]{algorithm2e}
\usepackage{setspace}
\usepackage{url}
\usepackage{fullpage}
\usepackage{makeidx}
\usepackage{enumerate}
\usepackage[top=1in, bottom=1.25in, left=1in, right=1in]{geometry}
\usepackage{graphicx,float,psfrag,epsfig,caption}
\usepackage[usenames,dvipsnames,svgnames,table]{xcolor}
\definecolor{darkgreen}{rgb}{0.0,0,0.9}
\usepackage[pagebackref,letterpaper=true,colorlinks=true,pdfpagemode=none,citecolor=OliveGreen,linkcolor=BrickRed,urlcolor=BrickRed,pdfstartview=FitH]{hyperref}
\usepackage{epstopdf}
\usepackage{color}
\usepackage{xr}
\usepackage{subfig}
\usepackage{caption}
\usepackage{graphicx}
\usepackage[utf8]{inputenc}

\usepackage{scalerel,stackengine}

\stackMath
\newcommand\reallywidehat[1]{%
\savestack{\tmpbox}{\stretchto{%
  \scaleto{%
    \scalerel*[\widthof{\ensuremath{#1}}]{\kern.1pt\mathchar"0362\kern.1pt}%
    {\rule{0ex}{\textheight}}
  }{\textheight}%
}{2.4ex}}%
\stackon[-6.9pt]{#1}{\tmpbox}%
}
\usepackage[mathscr]{euscript}

\DeclareSymbolFont{rsfs}{U}{rsfs}{m}{n}
\DeclareSymbolFontAlphabet{\mathscrsfs}{rsfs}

\numberwithin{equation}{section}

\newtheoremstyle{myexample} 
    {\topsep}                    
    {\topsep}                    
    {\rm }                   
    {}                           
    {\bf }                   
    {.}                          
    {.5em}                       
    {}  

\newtheoremstyle{myremark} 
    {\topsep}                    
    {\topsep}                    
    {\rm}                        
    {}                           
    {\bf}                        
    {.}                          
    {.5em}                       
    {}  

\definecolor{darkgreen}{rgb}{0.0, 0.5, 0.0}

\usepackage[capitalize,noabbrev]{cleveref}

\theoremstyle{plain}
\newtheorem{theorem}{Theorem}[section]
\newtheorem{proposition}[theorem]{Proposition}
\newtheorem{lemma}[theorem]{Lemma}
\newtheorem{corollary}[theorem]{Corollary}
\theoremstyle{definition}

\newtheorem{assumption}[theorem]{Assumption}
\theoremstyle{myremark}
\newtheorem{remark}{Remark}[section]
\theoremstyle{myexample}
\newtheorem{example}[remark]{Example}

\usepackage[textsize=tiny]{todonotes}

\input{commands.tex}

\title{Compressing Tabular Data via Latent Variable Estimation}

\author{Andrea Montanari${}^*$ \and Eric Weiner\thanks{Project N}}

\begin{document}

\title{Compressing Tabular Data via Latent Variable Estimation}

\author{Andrea Montanari${}^*$ \and Eric Weiner\thanks{Project N}}

\maketitle

\begin{abstract}
Data used for analytics
and machine learning often take the form of tables with categorical  entries. 
We introduce a family of lossless compression algorithms for such data that
proceed in four steps: $(i)$~Estimate latent variables associated to rows and columns; 
$(ii)$~Partition the table in blocks according to the row/column latents;
$(iii)$~Apply a sequential (e.g. Lempel-Ziv) coder to each of the
 blocks; $(iv)$~Append a compressed encoding of the latents.

We evaluate it on several benchmark datasets,
and study optimal compression in a probabilistic model for that tabular data, whereby
latent values are independent and table entries are conditionally independent given the latent 
values. We prove that the model has a well defined entropy rate 
 and satisfies an asymptotic equipartition property.
We also prove that classical compression schemes such as Lempel-Ziv and finite-state encoders
do not achieve this rate. On the other hand, the latent estimation strategy outlined
above achieves the optimal rate.
\end{abstract}

\section{Introduction}
\label{sec:Introduction}

Classical theory of lossless compression \cite{cover2006elements,salomon2004data} assumes 
that data take the form
of a random vector $\bX^N =(X_1,X_2,\dots,X_N)$ of length $N$ with entries in a finite alphabet
$\cX$. Under suitable ergodicity assumptions, the entropy per letter converges to a limit
$h := \lim_{N\to\infty}H(\bX^N)/N$ (Shannon-McMillan-Breiman theorem). Universal coding schemes
(e.g. Lempel-Ziv coding) do not requite knowledge of the distribution of $\bX^N$,
and can encode such a sequence without information loss using (asymptotically) $h$ bits per symbol.

While this theory is mathematically satisfying, its modeling assumptions
(stationarity, ergodicity) are unlikely to be satisfied in many applications. 
This has long been recognized by practitioners. 
The main objective of this paper is to investigate this fact mathematically in the context 
of tabular data, characterize the gap to optimality of classical schemes, and 
describe an asymptotically optimal algorithm that overcomes their limitations.

We consider a data table with $m$ rows and $n$ columns and entries in $\cX$,   
$\bX^{m,n}\in\cX^{m\times n}$  $\bX^{m,n} :=(X_{ij})_{i\le m,j\le n}$. 
The standard approach to such data is: $(i)$~Serialize, e.g. in row-first order, 
to form a vector of length $N=mn$, $\bX^N= (X_{11},X_{12},\dots , X_{1n}, X_{21},\dots, X_{mn})$;
$(ii)$~Apply a standard compressor (e.g., Lempel-Ziv) to this vector.

We will show, both empirically and mathematically that this standard approach can be 
suboptimal in the sense of not achieving the optimal compression rate.
This happens even in the limit of large tables, as long as the number of columns and rows 
are polynomially related (i.e. $n^{\eps}\le m\le n^{M}$ for some small constant $\eps$
and large constant $M$). 

We advocate an alternative approach:
\begin{enumerate}
\item Estimate row/column latents $\bu^m = (u_1,\dots,u_m)\in \cL^m$,
$\bv^n= (v_1,\dots,v_n)\in\cL^n$, with $\cL$ a finite alphabet.
\item Partition the table in blocks according to the row/column latents,
Namely, for $u,v\in\cL$, define
\begin{align}
\bX(u,v) = \vec\big(X_{ij}:\; u_i=u,v_j=v\big)\, .
\end{align}
where $\vec(\bM)$ denote the serialization of matrix $\bM$ (either row-wise or column-wise).
\item Apply a base compressor (generically denoted by $\zip_{\cX}:\cX^*\to \{0,1\}^*$) to each block 
$\bX(u,v)$
\begin{align}
\bz(u,v) = \zip_{\cX}(\bX(u,v))\, ,\;\;\;\; \forall u,v\in\cL\, .
\end{align}
\item Encode the row latents and column latents using a possibly different 
compressor $\zip_{\cL}:\cX^*\to \{0,1\}^*$, to get $\bz_{\row}= \zip_{\cL}(\bu)$,
$\bz_{\col}= \zip_{\cL}(\bv)$. Finally output the concatenation (denoted by $\oplus$)
\begin{align}
\enc(\bX^{m,n}) = {\sf header} \conc  \bz_{\row} \conc \bz_{\col}\conc 
\bigoplus_{u,v\in \cL} \bz(u,v) \,.\label{eq:LatentGen}
\end{align}
Here ${\sf header}$ is a header that 
contains encodings of the lengths of subsequent segments.
\end{enumerate}
Note that encoding the latents can in general lead to a suboptimal 
compression rate. While this can be remedied with techniques such as bits-back
coding, we observed in our applications that this yielded limited improvement.
Our analysis shows that the rate improvement afforded by bits-back coding is only significant 
in certain special regimes.
We refer to Sections \ref{sec:Math} and \ref{sec:Discussion} for further discussion.

The above description leaves several design choices undefined, namely:
$(a)$~The latents estimation procedure at point 1;
$(b)$~The base compressor $\zip_{\cX}$ for the blocks $\bX(u,v)$;
$(c)$~The base compressor  $\zip_{\cL}$ for the latents.

We will provide details for a specific implementation in Section \ref{sec:Implementation},
alongside empirical evaluation in Section \ref{sec:Empirical}.
Section \ref{sec:Generative} introduces a probabilistic model for the data $\bX^{m,n}$,
and Section \ref{sec:Math} establishes our main theoretical results:
standard compression schemes are suboptimal on this model, while the above latents-based approach 
is asymptotically optimal.
Finally we discuss extensions in Section \ref{sec:Discussion}. 

\subsection{Related work}

The use of latent variables is quite prevalent in compression methods based on machine
learning and probabilistic modeling.
Hinton and Zemel \cite{hinton1993autoencoders} introduced the idea that stochastically generated
codewords (e.g., random latents) can lead to minimum description lengths via bits back coding.
This idea was explicitly applied to lossless compression 
using arithmetic coding in  \cite{frey1996free}, and ANS coding in
\cite{townsend2019practical,townsend2019hilloc}.

Compression via low-rank approximation is closely-related to our latents-based 
approach and has been studied in the past. An incomplete list of contributions includes
\cite{cheng2005compression}  (numerical analysis),
\cite{li2010tensor} (hyperspectral imaging),
\cite{yuan2005projective,hou2015sparse} (image processing),
\cite{taylor2013lossless} (quantum chemistry), \cite{phan2020stable} (compressing the gradient for distributed optimization),
 \cite{chen2021drone} (large language models compression).

The present paper contributes to this line of work, but departs from it in a number
of ways. $(i)$~We study lossless compression while earlier work is mainly centered
on lossy compression. $(ii)$~Most of the papers in this literature do not
precisely quantify compression rate: they do not `count bits.'  $(iii)$~We show empirically
an improvement in terms of lossless compression rate over state of the art.

Another related area is network compression:  simple graphs can be viewed as 
matrices with entries in $\{0,1\}$. In 
the case of graph compression, one is  interested only in such  matrices
up to graph isomorphisms. The idea of reordering the nodes of the network  and
exploiting similarity between nodes has been investigated in this context,
see e.g. \cite{boldi2004webgraph,chierichetti2009compressing,lim2014slashburn,besta2018survey}
However, we are not aware of results analogous to ours in this literature. 

To the best of our knowledge, our work is the first to prove that classical lossless compression
techniques do not achieve the ideal compression rate under a probabilistic model 
for tabular data. We characterize this ideal rate as well as the one
achieved by classical compressors, and prove that latents estimation can be used to close this gap.

\subsection{Notations}

We generally use boldface for vectors and uppercase boldface for matrices, 
without making any typographic distinction between numbers and random variables. 
When useful, we indicate by superscripts the dimensions of a matrix or a vector:
$\bu^m$ is a vector of length $m$, and $\bX^{m,n}$ is a 
matrix of dimensions $m\times n$.
For a string $\bv$ and  $a\le b$, we use $\bv_a^{b}=(v_a,\dots,v_b)$ to denote the substring of $\bv$.

If $X,Y$ are random variables on a common probability space $(\Omega,\cF,\prob)$,
we denote by $H(X)$, $H(Y)$ their entropies, $H(X,Y)$ their joint entropy,
$H(X|Y)$ the conditional entropy of $X$ given $Y$. We will overload this notation:
if $p$ is a discrete probability distribution, we denote by $H(p)$ its entropy.
Unless stated otherwise, all entropies will be measured in bits.  For $\eps\in [0,1]$,
$\entro(\eps):=-\eps\log_2\eps-(1-\eps)\log_2(1-\eps)$.

\section{Implementation}
\label{sec:Implementation}

\subsection{Base compressors}

We implemented the following two options for the base compressors $\zip_{\cX}$ 
(for data blocks) and $\zip_{\cL}$ (for latents).

\paragraph{Dictionary-based compression (Lempel-Ziv, LZ).}
For this we used Zstandard (ZSTD) Python bindings to the  C implementation using
 the library {\sf zstd}, with level 12. While ZSTD can
 use run-length encoding schemes or literal encoding schemes,  we verified that in
  in this case ZSTD always use its LZ algorithm. 
  
The LZ algorithm in ZSTD is somewhat more sophisticated than the 
plain LZ algorithm used in our proofs. In particular it includes \cite{collet2018zstandard}
Huffman coding of literals 0-255 and entropy coding of the LZ stream.
Experiments with other (simpler) LZ implementations yielded similar results. 
We focus on ZSTD because of its broad adoption in industry.

\paragraph{Frequency-based entropy coding (ANS).} For each data portion
(i.e each block $\bX(u,v)$ and each of the row latents
$\bu$ and column latents $\bv$) compute empirical frequencies of the corresponding symbols.
Namely for all $u,v\in\cL$, $x\in\cX$, we compute 
\begin{align*}
&\hQ(x|u,v) := \frac{1}{N(u,v)}\sum_{i:u_i=u}\sum_{j:v_j=v}\bfone_{x_{ij}=x}\,,\\
&\hro(u) := \frac{1}{m}\sum_{i=1}^m\bfone_{u_{i}=u}\, , 
\;\; \hco(v) := \frac{1}{n}\sum_{i=1}^n\bfone_{v_{i}=v}\, , 
\end{align*}
where $N(u,v)$ is the number of $i\le m$, $j\le n$ such that $u_i=u$, $v_j = v$.
We then apply ANS coding \cite{duda2009asymmetric} to each block $\bX(u,v)$ modeling its entries as independent with
distribution $\hQ(\,\cdot\, |u,v)$, and to the latents $\bu^m$, $\bv^n$ using
the distributions $\hro(\,\cdot\,)$, $\hco(\,\cdot\,)$. We separately encode these 
counts as long integers.

Since our main objective was to study the impact of learning latents, we did not
try to optimize these base compressors.

\subsection{Latent estimation}
\label{sec:LatentEst}

We implemented latents estimation using a spectral clustering 
algorithm outlined in the pseudo-code above.

\vspace{0.2cm}
\begin{algorithm}[tb]
\caption{Spectral latents estimation}
\label{alg:spectral_est}
\begin{algorithmic}
\STATE {\bfseries Input:} {Data matrix $\bX^{m,n}\in\cX^{m\times n}$\\ latents range $k=|\cL|$; map $\psi:\cX\to\reals$}
\STATE {\bfseries Output:} {Factors $\bu^m\in \cL^m$, $\bv^n\in \cL^n$}
\STATE
\STATE Compute top $(k-1)$ singular vectors of $\bM^{m,n} = \psi(\bX^{m,n})$, $(\tba_i)_{i\le k-1}$,
$(\tbb_i)_{i\le k-1}$\;
\STATE Stack singular vectors in matrices $\bA = [\tba_1|\dots |\tba_{k-1}]\in\reals^{m\times (k-1)}$, $\bB = [\tbb_1|\cdots |\tbb_{k-1}]\in\reals^{n\times (k-1)}$;
\STATE Let $(\ba_i)_{i\le m}$, $\ba_i\in\reals^{k-1}$ be the rows of $\bA$;  $(\bb_i)_{i\le n}$,
$\bb\in\reals^{k-1}$ the rows of $\bB$\;

\STATE Apply KMeans to $(\ba_i)_{i\le m}$; store the cluster labels as vector $\bu^m$\;
\STATE Apply KMeans to $(\bb_i)_{i\le n}$; store the cluster labels as vector $\bv^n$\;
\STATE {\bfseries return} $\bu^m$, $\bv^n$

\end{algorithmic}
\end{algorithm}

A few remarks are in order. 
The algorithm encodes the data matrix $\bX^{m,n}$ as an $m\times n$ real-valued matrix
$\bM^{m,n}\in\reals^{m\times n}$ using a map  $\psi:\cX\to\reals$. In our experiments
we did not optimize this map and encoded the elements of $\cX$ as $0,1,\dots,|\cX|-1$
arbitrarily, cf. also Section \ref{sec:LatentProof}
 
The singular vector calculation turns out to be the most time consuming
part of the algorithm. Computing approximate singular vectors via power iteration
requires in this case of the order of  $\log(m\wedge n)$ matrix vector multiplications for each of $k$
vectors\footnote{This complexity assumes that the leading $k-1$ singular values are separated by a 
gap from the others. This is the regime in which the spectral clustering algorithm is successful.}. 
This amounts to $mnk\log(m\wedge n)$ operations, which is  larger
than the time needed to compress the blocks or to run KMeans. 
A substantial speed-up is obtained
via row subsampling, cf. Section \ref{sec:Discussion}

For the clustering step we used KMeans with $k$ clusters, initialized randomly.
More specifically, we use the {\sf scikit-learn} implementation  via 
{\sf sklearn.cluster.KMeans}. 
The overall latent estimation approach is quite basic and, in particular,
it does not try to estimate or make use of the model $Q(\,\cdot\,|u,v)$. 
\section{Empirical evaluation}
\label{sec:Empirical}

We evaluated our approach on tabular datasets with different origins.
Our objective is to assess the impact of using latents in reordering
columns and rows, so we will not attempt to achieve the best possible 
data reduction rate (DRR) on each dataset, but rather to compare compression 
with latents and without in as-uniform-as-possible fashion.

Since our focus is on categorical variables, we preprocess the data to fit
in this setting as described in Section \ref{sec:Preprocessing}.
This preprocessing step might involve dropping some of the columns of the original
table. We denote the number of columns after preprocessing by  $n$.

We point out two simple improvements we introduce in the implementation:
$(i)$~We use different sizes for rows latent alphabet and column latent alphabet
 $|\cL_r|\neq |\cL_c|$;
 $(ii)$~We choose  $|\cL_r|$, $|\cL_c|$ by optimizing  the compressed size .

\subsection{Datasets}

More details on these data can be found in Appendix \ref{app:Data}:

\noindent\emph{Taxicab.} A table with $m=62,495$, $n=18$ \cite{nycTLC2022taxi}. 
LZ: $|\cL_r| = 9$, $|\cL_c|= 15$. ANS: $|\cL_r| =5$, $|\cL_c|=14$.

\noindent\emph{Network.} Four social networks from \cite{snapnets} with
 $m=n\in\{333, 747, 786, 1187\}$. LZ and ANS: $|\cL_r| = 5$, $|\cL_c|= 5$.

\noindent\emph{Card transactions.}  A table with $m=24,386,900$  and 
$n=12$ \cite{ibm2019paper}. LZ and ANS: $|\cL_r| = 3$, $|\cL_c|= n$.

\noindent\emph{Business price index.}  A  table with 
$m=72,750$ and $n=10$  \cite{nz2022bpi}. LZ: $|\cL_r| = 6$, $|\cL_c|= 7$.
 ANS:  $|\cL_r| = 2$, $|\cL_c|=6$.

\noindent\emph{Forest.} A table from the UCI data repository with
$m=581,011$,  $n=55$  \cite{Dua:2019}.   LZ and ANS: $|\cL_r| = 6$, $|\cL_c|= 17$.

\noindent\emph{US Census.} Another table from \cite{Dua:2019} 
with $m=2,458,285$ and $n=68$.  LZ and ANS: $|\cL_r| = 9$, $|\cL_c|= 68$.

\noindent\emph{Jokes.} A collaborative filtering dataset with  $m=23,983$ rows and $n=101$
 \cite{goldberg2001paper,goldberg2001data}. 
 LZ: $|\cL_r| = 2$, $|\cL_c|= 101$.
 ANS:  $|\cL_r| = 8$, $|\cL_c|=8$.

\subsection{Results}

Given a lossless encoder $\phi:\cX^{m\times n}\to\{0,1\}^*$, we define its compression rate 
and data reduction rate (DRR) as
\begin{align}
\label{eq:RateDef}
\Rate_{\phi}(\bX^{m,n}):= \frac{\len(\phi(\bX^{m,n}))}{mn\log_2|\cX|}, \nonumber \\  
\;\;\;\;\; \;\;\DRR_{\phi}(\bX^{m,n}) := 1-\Rate_{\phi}(\bX^{m,n})\, .
\end{align}
(Larger DRR means better compression.)

The DRR of each algorithm is reported in Table \ref{table:DRR}. 
For the table of results, {\bf LZ} refers to row-major order ZSTD, {\bf LZ (c)} 
refers to column-major order ZSTD. We run
 KMeans on the data 5 times, with random initializations finding the DRR each time and 
 reporting the average. 

\begin{table*}[t]
\centering
\caption{Data reduction rate (DRR) achieved by classical and latent-based 
compressors on real tabular data.}\label{table:DRR}
\begin{tabular}{|c|c|c|c|c|c|c|}
\hline
{\bf Data} &  {\bf Size} &  {\bf LZ} &  {\bf LZ (c)} & {\bf ANS} & 
{\bf Latent $+$ LZ} &  {\bf Latent $+$ ANS}\\
\hline
\hline
Taxicab & 380 KB & 0.41 & 0.44 & 0.43 & 0.48  & $\boldsymbol{0.54}$\\
\hline
FB Network 1 & 13.6 KB & 0.63  & 0.63 & 0.76 & 0.58  & $\boldsymbol{0.78}$\\
\hline
FB Network 2 &  68.1 KB & 0.44 & 0.44 & 0.57 & 0.64  & $\boldsymbol{0.75}$\\
\hline
FB Network 3 & 75.4 KB & 0.59 & 0.59 & 0.75 & 0.69  & $\boldsymbol{0.80}$\\
\hline
GP Network 1 & 172 KB & 0.46 & 0.46 & 0.65 & 0.58  & $\boldsymbol{0.70}$\\
\hline
Forest (s) & 6.10 MB & 0.29 & 0.38 & 0.47 & 0.41  & $\boldsymbol{0.49}$\\
\hline
Card Transactions (s) & 123 MB & 0.03 & 0.21 &  0.29 & 0.20  & $\boldsymbol{0.30}$\\
\hline
Business price index (s) & 153 KB & $-0.03$ & 0.20 & 0.28 & 0.25  & $\boldsymbol{0.32}$\\
\hline
US Census & 43.9 MB & 0.38 & 0.31 & 0.47 & 0.52  & $\boldsymbol{0.62}$\\
\hline
Jokes & 515 KB & $-0.21$ & $-0.15$ & 0.07 & $-0.03$ & $\boldsymbol{0.14}$\\
\hline
\end{tabular}
\end{table*}

We make the following observations on the empirical results of 
Table \ref{table:DRR}. First, Latent $+$ ANS encoder achieves systematically the best DRR. Second,
the use of latent in several cases yields a DRR improvement of $5\%$ (of the uncompressed size)
or more. Third, as intuitively natural, this improvement appears to be larger for 
data with a large number of columns (e.g. the network data).

The analysis of the next section provides further support for these findings. 

\section{A probabilistic model}
\label{sec:Generative}

\begin{figure*}[t]
\centering
\includegraphics[width=0.8\textwidth]{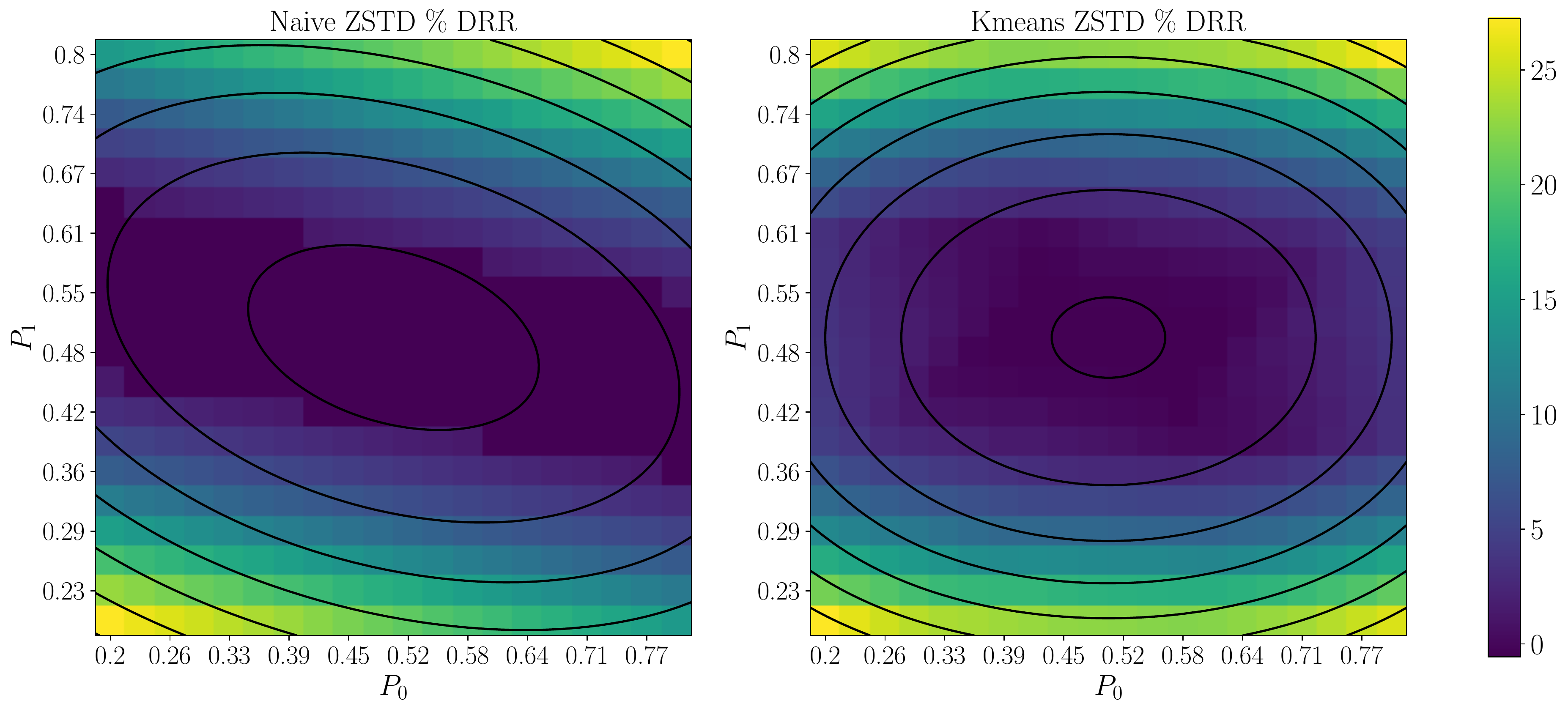} 
%
\includegraphics[width=0.8\textwidth]{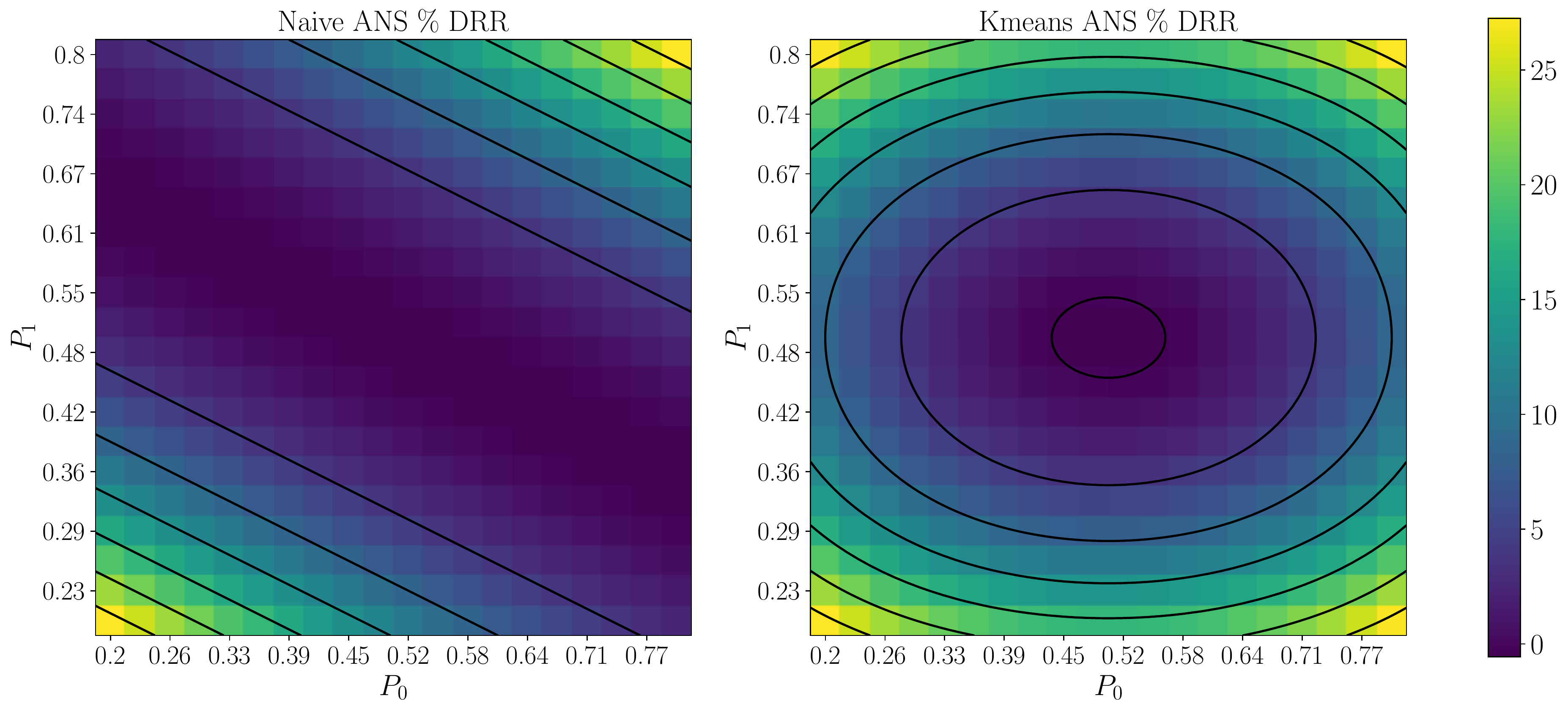} 
\caption{Comparing data reduction rate of naive  coding and latent-based 
coding for synthetically generated data. Top: ZSTD base compressor.
Bottom: ANS base compressor. Contour lines correspond to the 
compression rate predicted by the theorems of Section \ref{sec:Math} 
(coinciding with optimal rate for latent-based encoders).}\label{fig:Synthetic}
\end{figure*}
In order to better understand the limitations of classical approaches, and the
optimality of latent-based compression, we introduce a probabilistic model for
the table $\bX^{m,n}\in\cX^{m\times n}$.
We assume the true latents  $(u_i)_{i\le m}$, $(v_j)_{j\le n}$ 
to be independent random variables with 
\begin{align}
\prob(u_i = u)= \ro(u) \, ,\;\;\;\; \prob(v_i = v)= \co(v)\, .
\end{align}
We assume that the entries $(X_{ij})_{i\le m,j\le n}$ are conditionally independent given $\bu^m = (u_i)_{i\le m}$
$\bv^n = (v_j)_{j\le n}$, with 
\begin{align}
\prob\big(X_{ij}=x\big| \bu^m,  \bv^n\big) = Q(x|u_i, v_j)\, .
\end{align}
The distributions $\ro, \co$, and conditional distribution $Q$ are parameters of the model
(a total of $2(|\cL|-1)+|\cL|^2(|\cX|-1)$ real parameters). We will write
$(\bX^{m,n},\bu^m,\bv^n)\sim \cT(Q,\ro,\co;m,n)$ to indicate that the triple 
$(\bX^{m,n},\bu^m,\bv^n)$ is distributed
according to the model.
\begin{remark}
Some of our statements will be non-asymptotic, in which case  $m$, $n$, $\cX$, $\cL$, 
$Q$, $\ro$, $\co$ are fixed. Others will be of asymptotic. In the latter case, we have in mind
a sequence of problems indexed by $n$. In principle, we could write
$m_n$,  $\cX_n$, $\cL_n$, $Q_n$, $q_{\mbox{\small\rm r},n}$, $q_{\mbox{\small\rm c},n}$ to emphasize the fact that these quantities
depend on $n$.
However, we will typically omit these
subscripts.
\end{remark}

\begin{example}[Symmetric Binary Model]\label{ex:SBM}
As a toy example, we will use the following Symmetric Binary Model (SBM) 
which parallels the symmetric stochastic block model for community detection
\cite{holland1983stochastic}. 
We take $\cL = [k]:=\{1,\dots,k\}$, $\cX=\{0,1\}$, $\ro = \co= \Unif([k])$ (the uniform
distribution over $[k]$) and
\begin{align}
Q(1|u,v) = \begin{cases}
p_1 & \mbox{ if $u=v$,}\\
p_0  & \mbox{ if $u\neq v$.}\\
\end{cases}
\end{align}
We will write $(\bX^{m,n},\bu^m,\bv^n)\sim \cT_{\SBM}(p_0,p_1,k;m,n)$ when this distribution is
used. 

Figure \ref{fig:Synthetic} reports the results
of simulations within this model,  for ZSTD and ANS base compressors.
In this case $m=n=1000$, $k=3$, and we average DRR values over $4$ realizations.
Appendix \ref{app:Simulations} reports additional simulations under the same model for $k\in \{5,7\}$:
the results are very similar to the ones of Figure \ref{fig:Synthetic}.
As expected, the use of latents is irrelevant along the line $p_1\approx p_0$ (in this case,
the latents do not impact the distribution of $X_{ij}$). However, it becomes important when $p_1$ and $p_0$ 
are significantly different.

The figures also report contour lines of the theoretical predictions for the asymptotic
DRR of various compression algorithms (cf. Example \ref{example:SBM-rate}).
The agreement is excellent.
\end{example}

\section{Theoretical analysis}
\label{sec:Math}

In this section we present our theoretical results on compression rates
under the model $\cT(Q,\ro,\co,k;m,n)$ introduced above. 
We first characterize the optimal compression rate in Section \ref{sec:Ideal},
then prove that standard compression methods fail to attain this goal in 
Section \ref{sec:Failure}, and finally show that latent-based compression does in Section
 \ref{sec:LatentProof}. Proofs are deferred to Appendices
 \ref{sec:ProofIdeal}, \ref{sec:FiniteState}, \ref{sec:ProofLZ}, \ref{sec:ProofLatent}.

Throughout, we denote by $(X,U,V)$ a triple with joint 
 distribution $\prob(X=x,U=u,V=v) =Q(x|u,v)\ro(u)\co(v)$ (this is the same as the joint 
 distribution of $(X_{ij},u_i,v_j)$ for fixed $i,j$). 
 
 \subsection{Ideal compression}
 \label{sec:Ideal}
 
 Our first lemma provides upper and lower bounds on the entropy per symbol
 $H(\bX^{m,n})/mn$.
\begin{lemma}\label{lemma:EntropyTables}
Defining $H_{m,n}^+(X|U,V):= H(X|U,V) + 
 \frac{1}{n}H(U)+\frac{1}{m}H(V)$, 
we have 
\begin{equation}
 H(X|U,V) \le \frac{1}{mn}H(\bX^{m,n}) \le H^+_{m,n}(X|U,V)\, .
 \label{eq:SimpleEntroBD}
\end{equation}
Further, for any estimators $\hbu:\cX^{m\times n}\to\cL^m$, $\hbv:\cX^{m\times n}\to\cL^n$,
let  $\Acc_U := \min_{\pi\in \Perm_{\cL}}\sum_{i=1}^m \bfone_{\hu_i\neq \pi(u_i)}/m$,
 $\Acc_V := \min_{\pi\in \Perm_{\cL}}\sum_{i=1}^n \bfone_{\hv_i\neq \pi(v_i)}/n$
 ($\min$ over permutations of  $\cL$),
letting $\eps_U:=\E \Acc_U$, $\eps_V:=\E \Acc_V$,
we have
\begin{equation}
 H^+_{m,n}(X|U,V) - \delta_{m,n} \le  \frac{1}{mn}H(\bX^{m,n}) \le 
 H^+_{m,n}(X|U,V)\, . \label{eq:Fano}
\end{equation}
where $\delta_{m,n}:=\delta(\eps_U)/n+\delta(\eps_V)/m$ and
$\delta(\eps) :=\bentro(\eps)+\eps\log(|\cL|-1)$.
\end{lemma}

\begin{corollary} \label{coro:Ideal}
There exists
a lossless compressor $\phi$ whose rate (cf.Eq.~\eqref{eq:RateDef}) is
\begin{align}
\E\,\Rate_{\phi}(\bX^{m,n})\le \frac{1}{\log_2|\cX|} \Big\{ H^+_{m,n}(X|U,V) +\frac{1}{mn}\Big\}\, .\label{eq:IdealRate}
\end{align}
Further, for any lossless compressor $\phi$, 
$\E\,\Rate_{\phi}(\bX^{m,n})\ge H^+_{m,n}(X|U,V)  - \delta_{m,n}-2\log_2(mn)/mn$.
\end{corollary}

\begin{remark}
The simpler bound  \eqref{eq:SimpleEntroBD} implies that the entropy per entry is
$H(X|U,V)+O(1/(m\wedge n))$. The operational interpretation of this result is that we  
should be able to achieve the same compression rate per symbol \emph{as if} the latents were
given to us.

The additional terms  $\frac{1}{n}H(U)+\frac{1}{m}H(V)$ in Eq.~\eqref{eq:Fano}
account for the additional memory required for the latents. The lower bound in 
Eq.~\eqref{eq:Fano} implies that, if the latents can be accurately estimated from the data 
$\bX^{m,n}$ (that is if $\eps_U$, $\eps_V$ are small), then this overhead is essentially unavoidable.
\end{remark}

The nearly ideal compression rate in Eq.~\eqref{eq:IdealRate} can be achieved
by Huffmann or arithmetic coding, and requires knowledge of the probability 
distribution of $\bX^{m,n}$. Under the these schemes, the length of the codeword
associated to $\bX^{m,n}$ is within constant number of bits from $-\log_2\rP(\bX^{m,n})$,
where $\rP(\bX_0) := \prob(\bX^{m,n}=\bX_0)$ is the probability mass function of 
the random table $\bX^{m,n}$ \cite{cover2006elements,salomon2004data}.
The next lemma implies that the length concentrates tightly around the entropy. 
\begin{lemma}[Asymptotic Equipartition Property]\label{lemma:AEP}
For $\bX_0\in\cX^{m\times n}$, let  $\rP(\bX_0) = \rP_{Q,\ro,\co;m,n}(\bX_0)$ the probability 
of  $\bX^{m,n}=\bX_0$ under model  $\bX^{m,n}\sim\cT(Q,\ro,\co;m,n)$. Assume there exists a constant $c>0$
such that $\min_{x\in \cX}\min_{u,v\in\cL}Q(x|u,v)\ge c$. Then there exists a constant $C$ 
(depending on $c$) such that the 
following happens.

For $\bX^{m,n}\sim\cT(Q,\ro,\co;m,n)$ and any 
$t\ge 0$  with probability at least $1-2\, e^{-t}$:
\begin{align}
\big|-\log\rP(\bX^{m,n}) - H(\bX^{m,n})\big|\le  C\sqrt{mn(m+n)}\, t\, .
\end{align}
\end{lemma}
For the sake of simplicity, in the last statement we assume a uniform lower bound on 
 $Q(x|u,v)$. While such a lower bound holds without loss of generality
 when $Q$ is independent of $m,n$ (symbols with zero probability can be dropped), it might
  not hold in the $n$-dependent case. Appendix \ref{sec:ProofIdeal}
gives a more general  statement.

\subsection{Failure of classical compression schemes}
\label{sec:Failure}

We analyze two types of codes: finite-state encoders and Lempel-Ziv codes. 
Both operate on the 
serialized data $\bX^N=\vec(\bX^{m,n})$, $N=mn$, obtained by scanning the table in
row-first order (obviously column-first yields symmetric results). 

\subsubsection{Finite state encoders}

A finite state (FS) encoder takes the form of a triple $(\Sigma,f,g)$ with $\Sigma$ a finite set
of cardinality $M = |\Sigma|$ and 
$f:\cX\times \Sigma\to \{0,1\}^*$, $g:\cX\times \Sigma\to \Sigma$.

We assume that $\Sigma$ contains a special `initialization' symbol $s_{\sinit}$.
Starting from state $s_0=s_{\init}$, the encoder scans the input $\bX^N$ sequentially.
Assume after the first $\ell$ input symbols it is in state $s_\ell$, and produced 
encoding $\bz_1^{k(\ell)}$. 
Given input symbol $X_{\ell+1}$, it appends $f(X_{\ell+1},s_\ell)$ to the codeword, and
updates its state to $s_{\ell+1} = g(X_{\ell+1},s_\ell)$.

With an abuse of notation, denote by $f_{\ell}(\bX^{\ell},s_{\init})\in\{0,1\}^*$ the binary sequence
obtained by applying the finite state encoder to 
$\bX^\ell=(X_1,\dots,X_{\ell})$ 
We say that the FS encoder is information lossless  if for any $\ell\in\naturals$,
$\bX^\ell\mapsto f_\ell(\bX^{\ell},s_{\sinit})$ is injective.
\begin{theorem}\label{thm:FS}
Let $\bX=\bX^{m,n}\sim\cT(Q,\ro,\co;m,n)$ and $\phi:=(\Sigma,f,g)$ be an information lossless
 finite state  encoder. 
Define the corresponding compression rate $\Rate_{\phi}(\bX)$,
as per Eq.~\eqref{eq:RateDef}.
Assuming $m>10$, $|\Sigma|\ge |\cX|$, and $\log_2|\Sigma|\le n\log_2|\cX|/9$,
\begin{align}
\E\, \Rate_{\phi}(\bX)&\ge \frac{H(X|U)}{\log_2|\cX|}-
10 \sqrt{\frac{\log|\Sigma|}{n\log|\cX|}} \cdot\log(n\log|\Sigma|)\, .
\end{align}
\end{theorem}
\begin{remark}
The leading term of the above lower bound is $H(X|U)/\log_2|\cX|$. 
Since conditioning reduces entropy, this is strictly larger than the ideal rate
which is roughly $H(X|U,V)/\log_2|\cX|$, cf. Eq.~\eqref{eq:IdealRate}.

The next term is negligible provided  $\log|\Sigma| \ll n\log|\cX|$. 
This condition is easy to interpret: it amounts to say that the finite state machine 
does not have enough states to memorize a row of the table $\bX^{m,n}$.
\end{remark}

\subsubsection{Lempel-Ziv}
\label{sec:LZ}

The pseudocode of the Lempel-Ziv algorithm that we will analyze is given in Appendix
\ref{sec:ProofLZ}.

In words, after the first $k$ characters of the input have been parsed,
the encoder finds the longest string $\bX_{k}^{k+\ell-1}$ which appears in the past.
It then encodes a pointer to the position of the earlier appearance of
the string $T_k$, and its length $L_k$. If a simbol $X_k$ never appeared in the past, 
we use a special encoding, cf.  Appendix
\ref{sec:ProofLZ}.

We encode the pointer $T_k$ in plain binary using $\lceil\log_2(N+|\cX|)\rceil$
bits (note that $T_k\in \{-|\cX|+1,\dots, 1,\dots,N\}$), and $L_k$ using an instantaneous 
prefix-free code, e.g. Elias $\delta$-coding, taking $2\lfloor\log_2 L_k\rfloor+1$ bits.
%
\begin{assumption}\label{ass:NonDet}
There exist a constant $c_0>0$ such that 
\begin{equation*}
\max_{x\in\cX}\max_{u,v\in\cL}Q(x|u,v)\le 1-c_0\, .
\end{equation*}
Further  $Q,\ro,co, \cX, \cL$ are fixed and $m, n\to\infty$  with $m=n^{\alpha+o(1)}$, i.e. 
\begin{equation}
\lim_{n\to\infty}\frac{\log m}{\log n} = \alpha\in (0,\infty)\, .
\label{eq:AssPolynomial}
\end{equation}
\end{assumption}
As mentioned above, we consider sequences of instances with $m, n\to\infty$.
If convenient, the reader can think this sequence to be indexed by $n$, and let $m=m_n$
depend on $n$ such that Eq.~\eqref{eq:AssPolynomial} holds.

\begin{theorem}\label{thm:LempelZiv}
Under Assumption \ref{ass:NonDet},  the asymptotic Lempel-Ziv rate is
\begin{align}
\lim_{m,n\to\infty}&\E\, \Rate_{\LZ}(\bX^{m,n}) =
\Rate_{\LZ}^{\infty}:=
\sum_{u\in\cL}\frac{\ro(u) \Rate_{\LZ}^{\infty}(u)}{\log_2|\cX|}\, ,\label{eq:Asymp-LZ}\\
\Rate_{\LZ}^{\infty}(u)&:=
H(X|U=u) \wedge \Big(\frac{1+\alpha}{\alpha}\Big)H(X|U=u,V)\, .\nonumber
\end{align}
\end{theorem}

\begin{remark}
The asymptotics of the Lempel-Ziv rate is given by the minimum of two expressions,
which correspond to different behaviors of the encoder.
For $u\in\cL$, define $\alpha_*(u):= H(X|U=u,V)/(H(X|U=u)-H(X|U=u,V))$ (with $\alpha_*(u)=\infty$ if
$H(X|U=u)=H(X|U=u,V)$). Then:

 If $\alpha<\alpha_*(u)$, then we are a `skinny table' regime.
The algorithm mostly deduplicates segments in rows with latent $u$
by using strings in different rows but aligned in the same columns.
 If $\alpha>\alpha_*(u)$, then we are a `fat table' regime. The algorithm
 mostly deduplicates  segments on rows with latent $u$ by using 
 rows and columns that are not the same as the current segment.
\end{remark}

\begin{example}[Symmetric Binary Model, dense regime]\label{example:SBM-rate}
Under the Symmetric Binary Model  $\cT_{\SBM}(p_0,p_1,k;m,n)$ of Example \ref{ex:SBM},
we can compute the optimal compression rate of Corollary \ref{coro:Ideal},
the finite state compression rate of Theorem \ref{thm:FS}, the Lempel-Ziv
rate of Theorem \ref{thm:LempelZiv}. 

If $p_0$, $p_1$ are of order one, and $m=n^{\alpha+o_n(1)}$ as $m,n\to\infty$,
letting $\avp := ((k-1)/k)p_0+(1/k)\, p_1$,
$\overline{\entro}(p_0,p_1) := ((k-1)/k)\entro(p_0)+(1/k)\, \entro(p_1)$,
we obtain:
\begin{align*}
\E\, \Rate_{\mbox{\tiny\rm opt}}(\bX) &= 
\Big(1-\frac{1}{k}\Big)\, \entro(p_0)+\frac{1}{k}\, \entro(p_1) +o_n(1)\, ,\\
\E\, \Rate_{\mbox{\tiny\rm fin. st.}}(\bX) &\ge \entro(\avp) +o_n(1)\, ,\\
\E\, \Rate_{\LZ}(\bX)  &=  \entro(\avp)\wedge \Big(\frac{1+\alpha}{\alpha}\Big)\overline{\entro}(p_0,p_1)
+o_n(1)\, .
\end{align*}
These theoretical predictions are used to trace the contour lines in Figure \ref{fig:Synthetic}.
(ANS coding is implemented as a finite state code here.)
\end{example}

\subsection{Practical latent-based compression}
\label{sec:LatentProof}

Achieving the ideal rate of  Corollary \ref{coro:Ideal} via arithmetic or Huffmann coding
requires  to compute the probability $\rP(\bX^{m,n})$, which is intractable. 
We will next show that we can achieve a compression rate that is close to the
ideal rate via latents estimation.

We begin by considering general latents estimators  $\hbu: \cX^{m\times n}\to \cL^m$,  
$\hbv: \cX^{m\times n}\to \cL^n$. We measure their accuracy by the 
error  (cf. Lemma \ref{lemma:EntropyTables})
\begin{equation*}
\Acc_U(\bX;\hbu):= \frac{1}{m}
\min_{\pi\in \Perm_{\cL}}\sum_{i=1}^m \Big\{\bfone_{\hu_i(\bX)\neq \pi(u_i)}\Big\}\nonumber
\end{equation*}
and the analogous $\Acc_V(\bX;\hbv)$.
Here the minimization is over the set $\Perm_{\cL}$ of permutations of the latents alphabet $\cL$.

We can use any estimators $\hbu$, $\hbv$ to reorder rows and columns and compress
the table $\bX^{m,n}$ according to the algorithm described in the introduction.
We denote by  $\Rate_{\slat}(\bX)$ the compression rate achieved by
such a procedure.

Our first result implies that, if the latent estimators are 
consistent (namely, they recover the true latents with high probability, up to permutations),
then the resulting rate is close to the ideal one.
\begin{lemma}\label{lemma:Consistency}
Assume data distributed according to model $\bX^{m,n}\sim \cT(Q,\ro,\co;m,n)$,
with $m,n\ge \log_2|\cL|$.
Further assume there exists $c_0>0$ such that $\ro(u),\co(v)\ge c_0$ for all $u,v\in\cL$.
Let $\Rate_{\slat}(\bX)$ be the rate achieved by the latent-based scheme
with latents estimators $\hbu$, $\hbv$, and base encoders $\zip_{\cX}=
\zip_{\cL}=\zip$. Then 
\begin{align}
\E\,\Rate_{\slat}&(\bX)\le \frac{H(\bX^{m,n})}{mn\log_{2}|\cX|}+2P_{\mbox{\tiny\rm err}}(m,n)
+\frac{4\log(mn)}{mn}\nonumber \\
&+
 |\cL|^2\Delta_{\zip}(c\cdot mn;\cuQ)+2\Delta_{\zip}(m\wedge n;\{\ro,\co\})\, .
 \label{eq:GeneralUB-Consistency}
\end{align}
Here $P_{\mbox{\tiny\rm err}}(m,n):=\prob( \Acc_U(\bX^{m,n};\hbu)>0) +
\prob( \Acc_V(\bX^{m,n};\hbv)>0)$, $\Delta_{\zip}(N;\cuP_*)$ is the worst-case redundancy of encoder $\zip$ over i.i.d. sources with
distributions in $\cuP_*$ (see comments below),
$\cuQ:=\{Q(\,\cdot\,|u,v)\}_{u,v\in\cL}$.

The redundancies of Lempel-Ziv, frequency-based arithmetic coding and 
ANS coding can be upper bounded  as (in the last bound 
$Q,\ro,\co$ need to be be independent of $N$)
\begin{align}
\Delta_{\LZ}(N;\cuP_*) & \le 40 c_*(\cuP_*)\Big(\frac{\log \log N}{\log N}
\Big)^{1/2}\, , \label{eq:LZ-overhead}\\ 
\Delta_{\AC }(N;\cuP_*) & \le  \frac{2|\cX|}{\log|\cX|}
\cdot \frac{\log N}{N}\, ,\label{eq:AC-overhead}\\
\Delta_{\ANS}(N;\cuP_*) & \le   \frac{2|\cX|\log N+C_{|\cX|}}{N}\, . \label{eq:ANS-overhead}
\end{align}
Here Eq.~\eqref{eq:LZ-overhead} holds for 
$N\ge \exp\{\sup_{q\in \cuP_*}(4\log (2/H(q)))^2\}$, and 
$c_*(\cuP_*) := \sup_{q\in \cuP_*}\sum_{x\in\cX}(\log q(x))^2/|\cX|$. 
\end{lemma}
The proof of this lemma is given in Appendix \ref{sec:ProofConsistency}.
The main content of the lemma is in the general bound \eqref{eq:GeneralUB-Consistency}
which is proven in Appendix \ref{sec:ProofGeneralUBConsistency}.
\begin{remark} 
We define the worst case redundancy $\Delta_{\zip}(N_0;\cuP_*):=
\max_{N\ge N_0}\widehat{\Delta}_{\zip}(N;\cuP_*)$, where
\begin{align}
\widehat{\Delta}_{\zip}(N;\cuP_*):= \max_{q\in \cuP_*} \frac{\E_q\len(\zip(\bY^N))- H(\bY^N)}{N\log_2 k}
\,,\label{eq:DefOverhead}
\end{align}
where $\cuP_*\subseteq \cuP([k]):=\{(p_i)_{i\le k}\in\reals^k: p_i\ge 0\,\forall i\, \mbox{ and }\sum_{i\le k}p_i=1\}$ 
is a set of
probability distributions over $[k]$ and $\bY^N$ is a vector with i.i.d.
entries $Y_i\sim q$. 

While Eqs.~\eqref{eq:LZ-overhead}---\eqref{eq:ANS-overhead} are 
closely related to well known facts, there are nevertheless differences with respect to 
statements in the literature. We address them in Section \ref{sec:Redundancy}.
Perhaps the most noteworthy difference is in the bound \eqref{eq:LZ-overhead} for the 
LZ algorithm. Existing results, e.g. Theorem 2 in \cite{savari1998redundancy},
assume  a single, $N$-independent, distribution $q$ and are asymptotic in nature. 
Equation \eqref{eq:LZ-overhead} is a non-asymptotic statement and applies to 
a collection of distributions $\cuP_*$ that could depend on $N$.  
\end{remark}

Lemma \ref{lemma:Consistency} can be used in conjunction with any latent estimation
algorithm, as we next demonstrate by considering the spectral algorithm
of Section \ref{sec:LatentEst}. Recall that the algorithm makes use of a map $\psi:\cX\to\reals$.
For $(X,U,V)\sim Q(\,\cdot\,|\,\cdot\,,\,\cdot\,)\ro(\,\cdot\,)\co(\,\cdot\,)$,
we define   $\opsi(u,v) := \E[\psi(X)|U=u,V=v]$,
$\bPsi:=\big(\opsi(u,v)\big)_{u,v\in\cL}$ and the parameters:
\begin{align}
&\mu_n:= \sigma_{\min}(\bPsi)\, ,\;\; 
\nu_n: = \max_{u,v\in \cL}|\opsi(u,v)|\, ,\\
&\sigma_n^2 := \max_{u,v\in \cL} {\rm Var}\big(\psi(X)|U=u,V=v\big)\, .
\end{align}
We further will assume, without loss of generality $\max_{x\in\cX}|\psi(x)|\le 1$.

Finally, we need to formally specify the version of the KMeans primitive in the spectral 
clustering algorithm. In fact, we establish correctness for a simpler thresholding procedure.
Considering to be definite the row latents, and for a given threshold $\theta>0$,
we construct a graph $G_{\theta}=([m],E_{\theta})$ by letting (for distinct $i,j\in [m]$)
\begin{align}
\frac{\|\ba_i-\ba_j\|_2}{(\|\ba_i\|_2+\|\ba_j\|_2)/2}\le \theta \;\;
\Leftrightarrow \;\; (i,j)\in E_{\theta}\, .\label{eq:Thresholding}
\end{align}
The algorithm then output the connected components of $G_{\theta}$.
\begin{theorem}\label{thm:MainLatent}
Assume data  $\bX^{m,n}\sim \cT(Q,\ro,\co;m,n)$, with $m,n\ge \log_2|\cL|$
and $\min_{u\in \cL}(\ro(u)\wedge\co(u))\ge c_0$ for a constant $c_0>0$.
Let $\Rate_{\slat}(\bX)$ be the rate achieved by the latent-based scheme
with spectral latents estimators $\hbu$, $\hbv$,  base compressors $\zip_{\cX}=
\zip_{\cL}=\zip$, and thresholding algorithm as described above. Then, assuming $\sigma_n\ge c\sqrt{(\log n)/n}$, $\nu_n/\sigma_n\le c\sqrt{\log n}$ 
$\mu_n \ge C(\sigma_n\sqrt{(\log n)/m}\vee (\log n)/\sqrt{mn})$, $\theta\le \sqrt{c_0}/100$,
we have 
\begin{align*}
\E\,\Rate_{\slat}&(\bX)\le \frac{H(\bX^{m,n})}{mn\log_{2}|\cX|}+\frac{10\log(mn)}{mn} \\
&+
 |\cL|^2\Delta_{\zip}(c\cdot mn;\cuQ)+2\Delta_{\zip}(m\wedge n;\{\ro,\co\})\, .
\end{align*}
\end{theorem}

We focus on the simpler thresholding algorithm of Eq.~\eqref{eq:Thresholding}  
instead of KMeans in order to avoid technical complications that are not the main focus of 
this paper. We expect it to be relatively easy to generalize this result, e.g. using
the results of  \cite{makarychev2020improved} for KMeans++.

\begin{example}
Consider the Symmetric Binary Model  $\cT_{\SBM}(p_0,p_1,k;m,n)$ of Example \ref{ex:SBM},
with $p_0=p_{0,n}$, $p_1=p_{1,n}$ potentially dependent on $n$. Since in this case $\cX=\{0,1\}$
the choice of the map $\psi$ has little impact and we set $\psi(x)=x$. 
We assume, to simplify formulas, $|p_{1,n}-p_{0,n}|\le kp_{0,n}$, $p_{1,n}\vee p_{0,n}\le 9/10$. It is easy to compute 
$\mu_n=|p_{1,n}-p_{0,n}|$, $\nu_n= p_{0,n}\vee p_{1,n}$, $\sigma^2_n\asymp p_{1,n}\vee p_{0,n}$:

Theorem \ref{thm:MainLatent} implies nearly optimal compression rate under the 
following conditions on the model parameters:
\begin{align*}
&p_{1,n}\vee p_{0,n} \gtrsim \sqrt{\frac{\log n}{n}}\, ,\;\;\;\;\;|p_{1,n}-p_{0,n}|\gtrsim
\frac{\log n}{\sqrt{mn}}\, ,\\
&\frac{|p_{1,n}-p_{0,n}|}{p_{1,n}\vee p_{0,n}}\gtrsim \sqrt{\frac{\log n}{n}}\, .
\end{align*}
Here $\gtrsim$ hides factors depending on $k,c_0$. The last of these condition amounts to 
requiring that the signal-to-noise ratio is large enough to consistently reconstruct the latents.
In the special case of square symmetric matrices ($m=n$), 
sharp constants in these bounds can be derived from \cite{abbe2017community}.
\end{example}

\section{Discussion and extensions}
\label{sec:Discussion}

We proved that classical lossless compression schemes, that serialize the data and then 
apply a finite state encoder or a Lempel-Ziv encoder to the resulting sequence are 
sub-optimal when applied to tabular data. Namely, we introduced a simple model for tabular data, and
  made the following 
novel contributions:

\noindent{\emph 1.} We characterized the optimal compression rate under this model.

\noindent{\emph 2.} We rigorously quantified the gap in compression rate suffered by
classical compressors.

\noindent{\emph 3.}  We showed that a compression scheme that estimates the latents performs well in practice
and provably achieves optimal rate on our model.

The present work naturally suggests several extensions.

\noindent \emph{Faster spectral clustering via randomized linear algebra.} We implemented
row subsampling singular value decomposition \cite{drineas2006fast}, and observed hardly no loss 
in DRR by using $10\%$ of the rows. 

\noindent \emph{Bits back coding.}
As mentioned several times, encoding the latents is sub-optimal, unless these
can be estimated accurately in the sense of $\E\Acc_U(\bX;\hbu)\to 0$, 
$\E\Acc_V(\bX;\hbv)\to 0$, cf. Lemma \ref{lemma:EntropyTables}. 
If this is not the case, then optimal rates can be achieved using bits-back coding.

\noindent \emph{Continuous latents.} Of course, using discrete latents is somewhat un-natural, 
and it would be interesting to consider continuous ones, in which case 
bits-back coding is required.

\section*{Acknowledgements}
We are grateful to  Joseph Gardi, Evan Gunter, Marc Laugharn, Eren Sasoglu
for several conversations about this work.
This work was carried out while Andrea Montanari
was on leave from Stanford and a Chief Scientist at Ndata Inc dba Project N. The present
research is unrelated to AM’s Stanford activity.

\newpage

\bibliographystyle{amsalpha}
\newcommand{\etalchar}[1]{$^{#1}$}
\providecommand{\bysame}{\leavevmode\hbox to3em{\hrulefill}\thinspace}
\providecommand{\MR}{\relax\ifhmode\unskip\space\fi MR }
\providecommand{\MRhref}[2]{%
  \href{http://www.ams.org/mathscinet-getitem?mr=#1}{#2}
}
\providecommand{\href}[2]{#2}

\newpage
\appendix
\onecolumn

\section{Details on the empirical evaluation}

\subsection{Datasets}
\label{app:Data}

We used the following datasets:
\begin{itemize}
\item \emph{Taxicab.} A table with $m=62,495$
rows, $n_0=20$ columns comprising data for taxi rides in NYC during January 2022
\cite{nycTLC2022taxi}. After preprocessing this table has $n=18$ columns. For the LZ (ZSTD) 
compressor we 
used $|\cL_r| = 9$ row latents and $15$ column latents, for the ANS compressor we used $|\cL_r| =5$ row latents and 14 column latents.
\item \emph{Network.} Four social networks from SNAP Datasets, representing either friends 
as undirected edges for Facebook or directed following relationships on Google Plus \cite{snapnets}. 
We regard these as four distinct tables with $0-1$ entries, with dimensions, respectively 
 $m=n\in\{333, 747, 786, 1187\}$. For each table we used $5$ row latents and $5$ column latents.
\item \emph{Card transactions.}  A table of simulated credit card transactions containing 
information like card ID, merchant city, zip, etc. This table has $m=24,386,900$ rows and 
$n_0=15$ columns and was generated as described in \cite{ibm2019paper} and downloaded 
from \cite{ibm2021cardtransactions}. After preprocessing the table has $n=12$ columns. 
For this table we used $3$ row latents and $n$ column latents.
\item \emph{Business price index.}  A table of the values of the consumer price index 
of various goods in New Zealand between 1996 and 2022. This is a table with 
$m=72,750$ rows and $n_0=12$ columns from the Business price indexes: 
March 2022 quarter - CSV file from \cite{nz2022bpi}. After preprocessing this table has 
$n=10$ columns. Due to the highly correlated nature of consecutive rows, 
we first shuffle them before compressing. For the LZ method we used $6$ 
row latents and $7$ column latents, for the ANS method we used $2$ row latents and $6$ column latents.
\item \emph{Forest.} A table from the UCI data repository comprising 
$m=581,011$  cartographic measurements with $n_0=55$ attributes, to predict forest cover
 type based on information gathered from US Geological Survey \cite{Dua:2019}. It contains binary 
 qualitative variables, and some continuous values like elevation and slope. 
 After preprocessing this data has $n=55$ columns. For the LZ method we used $6$
  row latents and $17$ column latents, for the ANS method we used $6$ row latents and 17 column latents.
\item \emph{US Census.} Another table from the UCI Machine Learning Repository \cite{Dua:2019} 
with $m=2,458,285$ and $n_0=68$ categorical attributes related to demographic information, 
income, and occupation information. After preprocessing this data has 
$n=68$ columns. For this data we used $9$ row latents and $n$ column latents.
\item \emph{Jokes.} A table containing ratings of a series of jokes by 24,983 users 
collected between April 1999 and May 2003  \cite{goldberg2001paper,goldberg2001data}. These ratings are real numbers on a scale from $-10$ to $10$, and a value 
of 99 is given to jokes that were not rated. There are $m=23,983$ rows and $n_0=101$. 
The first column identifies how many jokes were rated by a user, and the rest of the 
columns contain the ratings. After preprocessing this data has $n=101$ columns, all quantized. 
For the LZ method we used $2$ row latents and $n$ column latents, for the ANS method we used  $8$ row latents and $8$ column latents. 
\end{itemize}

\subsection{Preprocessing}
\label{sec:Preprocessing}

We preprocessed different columns as follows:
\begin{itemize}
\item If a column comprises $K\le 256$ unique values, then we map the values
to $\{0,\dots,K-1\}$.
\item If a column is numerical and comprises more than $256$ unique values, 
we calculate the quartiles for the data and map each entry to its quartile
 membership (0 for the lowest quartile, 1 for the next largest, 2 for the next and 3 for the largest).
\item If a column does not meet either of the above criteria, we discard it.
\end{itemize}
Finally, in some experiments we randomly permuted before compression.
The rationale is that some of the above datasets have rows already ordered in a 
way that makes nearby rows highly correlated. In these cases, row reordering is --obviously--
of limited use.

\section{Further simulations}
\label{app:Simulations}

\begin{figure}[t]
\centering
\includegraphics[width=0.8\textwidth]{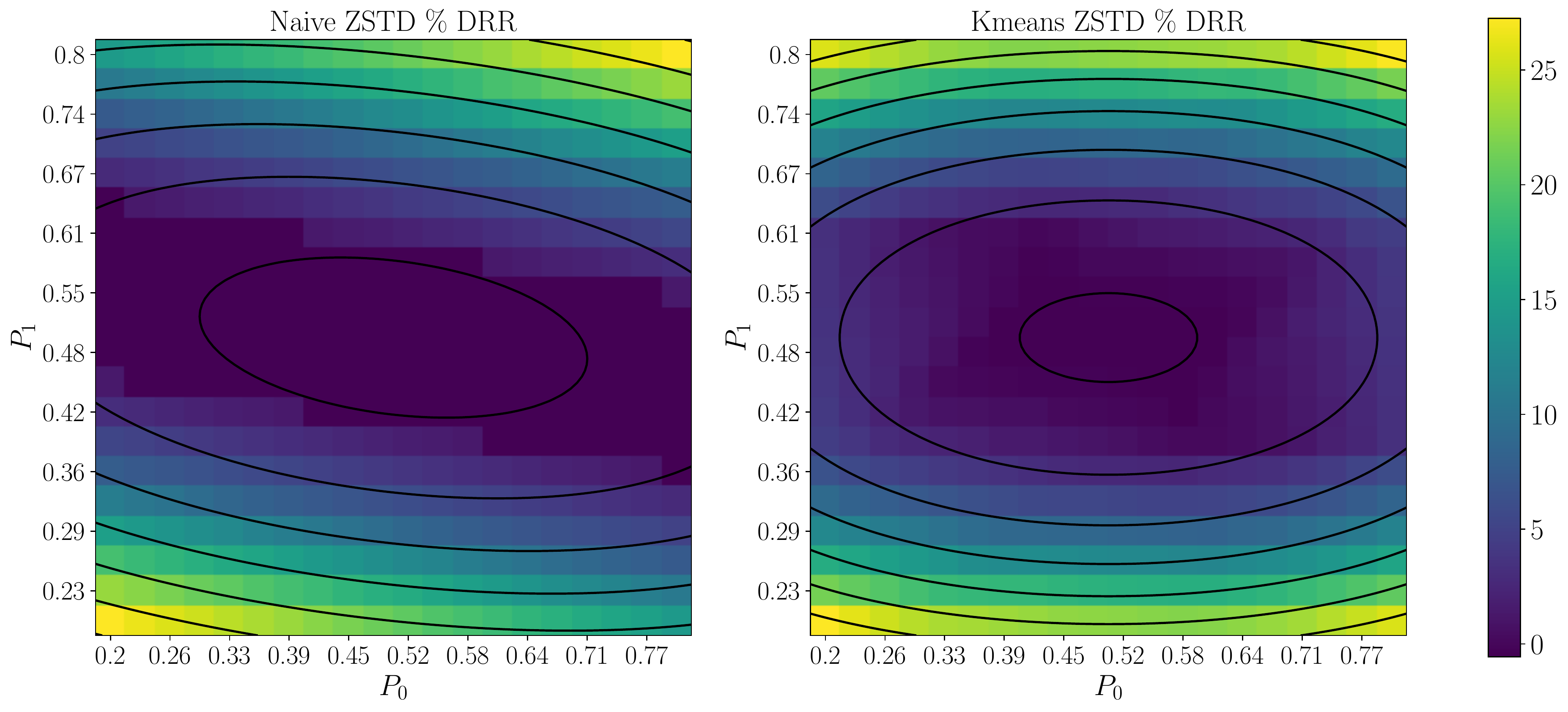} 
\includegraphics[width=0.8\textwidth]{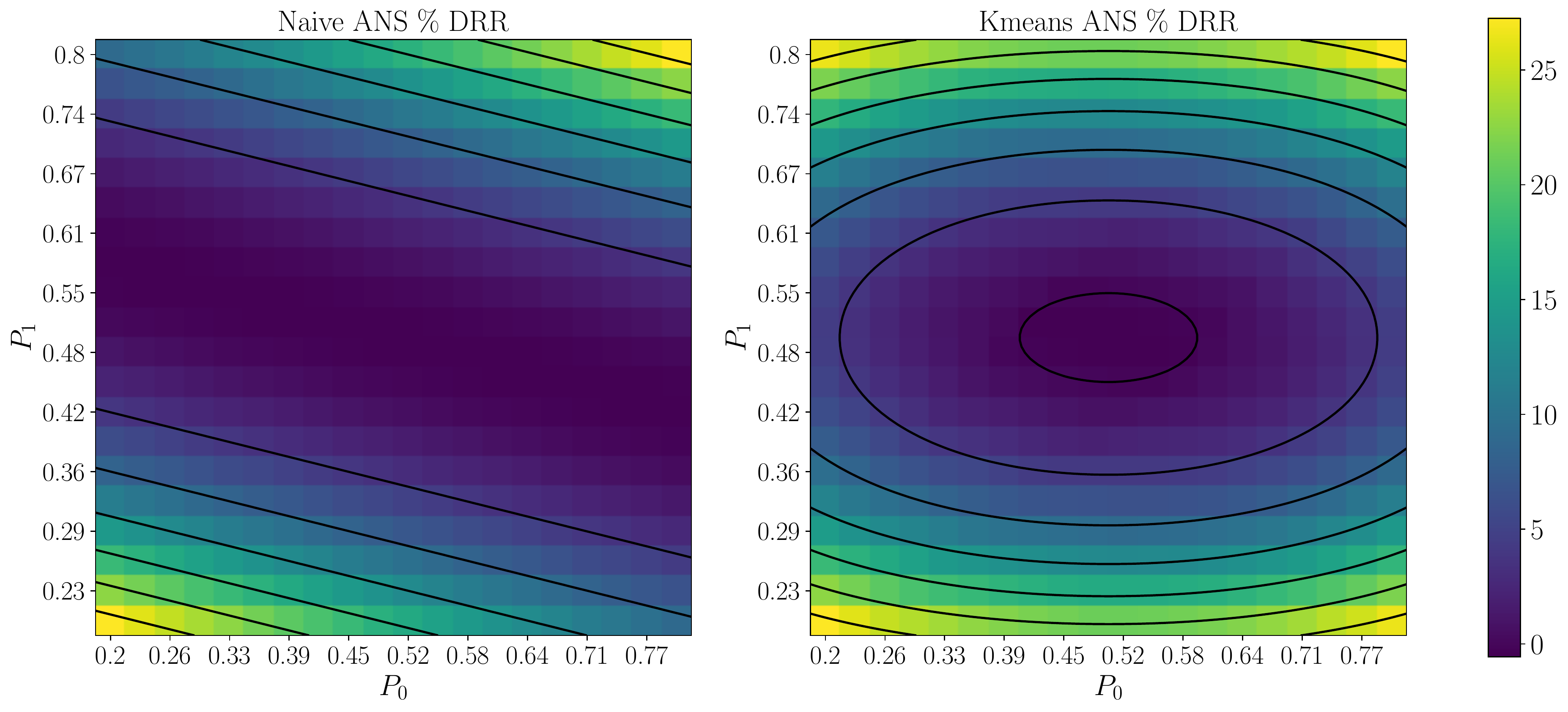} 
\caption{Comparing data reduction rate of naive coding and latent-based 
coding for data from SBM with $k=5$ latents. Top row: ZSTD base compressor.
Bottom row: ANDS base compressor. Contour lines correspond to the 
theoretical predictions for various compression algorithms (cf. Example~\ref{example:SBM-rate}).}\label{fig:Synthetic-5}
\end{figure}

\begin{figure}[t]
\centering
\includegraphics[width=0.8\textwidth]{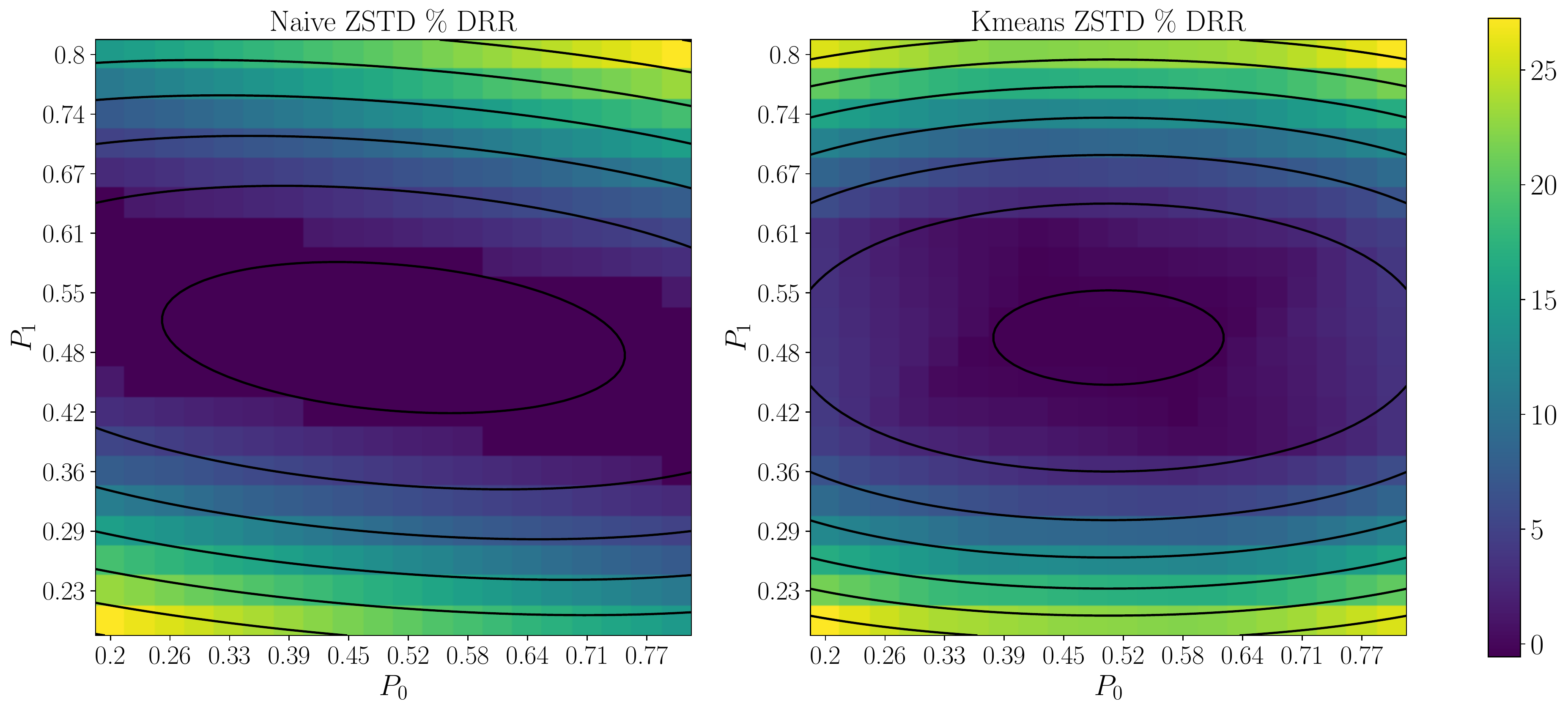} 
\includegraphics[width=0.8\textwidth]{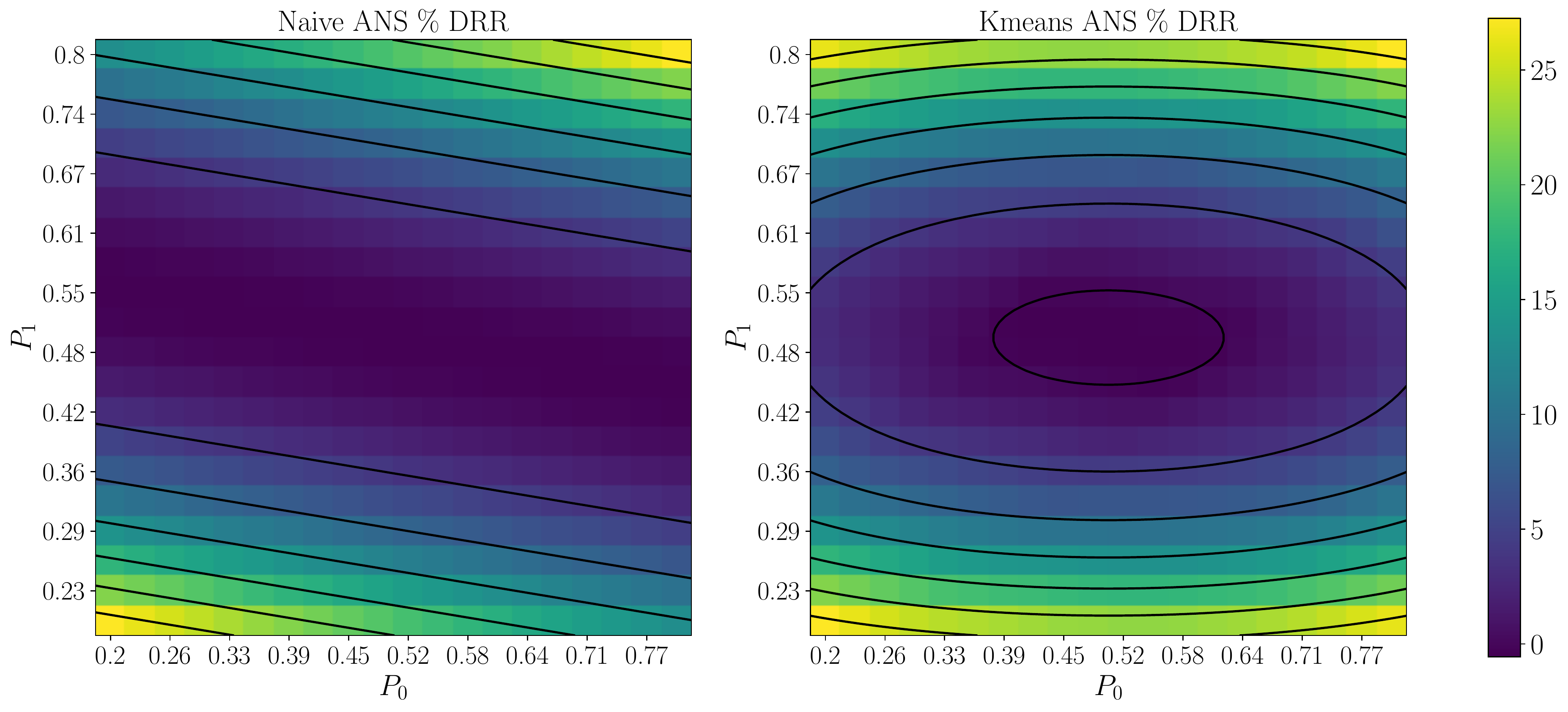} 
\caption{Same as Figure \ref{fig:Synthetic-5} with $k=7$.}\label{fig:Synthetic-7}
\end{figure}

Figures \ref{fig:Synthetic-5} and \ref{fig:Synthetic-7} report empirical
DRR values for ZSTD and ANS coding, for data generated according to the symmetric
model $\cT_{\SBM}(p_0,p_1,k;m,n)$ of Section \ref{sec:Generative}.
We use $m=n=1000$ as before, but now $k\in\{5,7\}$. Results confirm the conclusions of 
 Section \ref{sec:Generative}.

\section{A basic fact}

\begin{lemma}\label{lemma:SourceCodingConverse}
Let $\cA$ be a finite set and $F:\cA\to \{0,1\}^*$ be an injective map.
Then, for any probability distribution $p$ over $\cA$,
\begin{align}
\sum_{a\in \cA}p(a)\len(F(a))\ge H(p) -\log_2\log_2(|\cA|+2)\, .\label{eq:LB-Entropy}
\end{align}
\end{lemma}
\begin{proof}
Assume without loss of generality that $\cA = \{1,\dots, M\}$, with $|\cA|=K$, and that the elements of $\cA$
have all non-vanishing probability and are ordered by decreasing probability $p_1\ge p_2\ge \cdots\ge p_K>0$.
Let $N_\ell := 2+4+\dots+2^\ell=2{\ell+1}-2$. 
Then the expected length is minimized any map $F$ such that $\len(F(a))= \ell$ for $N_{\ell-1}\le a\le N_{\ell}$
with the maximum length $\ell_K$ being defined by $N_{\ell_K-1}< K \le N_{\ell_K}$. For $A\sim p$, 
$L:=\len(F(A))$, we have
\begin{align*}
H(p) &:= H(A) \stackrel{(a)}{\le} H(L) +H(A|L) \\
& \le \log_2 \ell_K + \sum_{\ell=1}^{\ell_K}\prob(L=\ell) H(A|L=\ell)\\
&\stackrel{(b)}{\le} \log_2 \ell_K + \sum_{\ell=1}^{\ell_M}\prob(L=\ell) \, \ell\\
&\le \log_2 \log_2(K+2)+  \sum_{a\in \cA}p(a)\len(F(a))\, ,
\end{align*}
where $(a)$ is the chain rule of entropy and $(b)$ follows because by injectivity,
given $\len(F(A))=\ell$, $A$ can  take at most $2^{\ell}$ values.
\end{proof}

\section{Proofs of results on ideal compression}
\label{sec:ProofIdeal}

\subsection{Proof of Lemma \ref{lemma:EntropyTables}}

We begin by claiming that 
\begin{align}
\frac{1}{mn}H(\bX^{m,n}) = H(X|U,V) +\frac{1}{mn}I(\bX^{m,n};\bU^m,\bV^n)\, .
\label{eq:Entro_MI}
\end{align}
Indeed, by the definition of mutual information, we have 
$H(\bX_{m,n}) = H(\bX_{m,n}|\bU_m,\bV_n) +I(\bX_{m,n};\bU_m,\bV_n)$. Equation 
\eqref{eq:Entro_MI} follows by noting that
\begin{align*}
 H(\bX^{m,n}|\bU^m,\bV^n) &= \sum_{\bu\in\cL^m} \sum_{\bv\in\cL^n}
\prob(\bU^m=\bu,\bV^n=\bv)\, H(\bX^{m,n}|\bU^m= \bu,\bV^n=\bv)\\
&\stackrel{(a)}= \sum_{i=1}^m\sum_{j=1}^n\sum_{\bu\in\cL^m} \sum_{\bv\in\cL^n}
\prob(\bU^m=\bu,\bV^n=\bv)\, H(X_{i,j}|U_i=u,V_j=v)\\
&= \sum_{i=1}^m\sum_{j=1}^n\, 
 \sum_{\bu\in\cL} \sum_{\bv\in\cL} \prob(U_i=u_i,V_j=v_j)H(X_{i,j}|U_i=u_i,V_j=v_j)\\
 &=\sum_{i=1}^m\sum_{j=1}^n H(X_{i,j}|U_i,V_j)\stackrel{(b)}{=}mn H(X_{1,1}|U_1,V_1)\, ,
 \end{align*}
 where $(a)$ follows from the fact that the $(X_{i,j})$ are conditionally independent given 
 $\bU^m$, $\bV^n$, ad since the conditional distribution of $X_{i,j}$
 only depends on $\bU^m$, $\bV^n$  via $U_i$, $V_j$; $(b)$ holds because 
 the triples $(X_{i,j},U_i,V_j)$ are identically distributed.
 
The lower bound in Eq.~\eqref{eq:SimpleEntroBD} holds because mutual information is non-negative,
and the upper bound because $I(\bX^{m,n};\bU^m,\bV^n)\le H(\bU^m,\bV^n)= m H(U_1)+n H(V_1)$.

Finally, to prove Eq.~\eqref{eq:Fano}, define
\begin{align}
\pi_{\bU,\bX}:= \arg\min_{\pi\in \Perm_{\cL}}\frac{1}{m}\sum_{i=1}^m \bfone_{\hu_i\neq \pi(u_i)}\, ,
\;\;\;\; \pi_{\bV,\bX}:= \arg\min_{\pi\in \Perm_{\cL}}\frac{1}{n}\sum_{i=1}^n \bfone_{\hv_i\neq \pi(v_i)}\, ,
\end{align}
If the minimizer is not unique, one can be chosen arbitrarily. 
We then have holds because
\begin{align}
I(\bX^{m,n};\bU^m,\bV^n) &= H(\bU^m,\bV^n)- H(\bU^m,\bV^n|\bX^{m,n})\nonumber\\
& \ge  mH(U_1)+nH(V_1) -H(\bU^m|\bX^{m,n})-H(\bV^n|\bX^{m,n})\nonumber\\
&\ge  mH(U_1)+nH(V_1) -\big[H(\bU^m|\bX^{m,n},\pi_{\bU,\bX})+I(\bU^m;\pi_{\bU,\bX}|\bX^{m,n})\big]\\
&\phantom{AAAA}-\big[H(\bV|\bX^{m,n},\pi_{\bV,\bX})+I(\bV^n;\pi_{\bV,\bX}|\bX^{m,n})\big]\nonumber\\
& \ge mH(U_1)+nH(V_1) -H(\bU^m|\bX^{m,n},\pi_{\bU,\bX})-H(\bV|\bX^{m,n},\pi_{\bV,\bX})-2|\cL|\log_2(|\cL|)
\, ,\label{eq:InfoClaim}
\end{align}
where in the last inequality we used the fact that
$I(\bU^m;\pi_{\bU,\bX}|\bX^{m,n})\le H(\pi_{\bU,\bX})\le \log_2(|\cL|!)$.

Now, consider the term $H(\bU^m|\bX^{m,n},\pi_{\bU,\bX})$. Letting $\bY:=(\bX^{m,n},\pi_{\bU,\bX})$
the stated accuracy assumption implies that there exists an estimators
$\hbu^+ = \hbu^+(\bY)$ such that 
\begin{align*}
\eps_U=\frac{1}{m}\sum_{i=1}^m \eps_{U,i}\, ,\;\;\;\;\;\eps_{U,i}:= \prob\big(\hu^+_i(\bY)\neq U_i\big)
\, .
\end{align*}
By Fano's inequality
\begin{align*}
H(\bU^m|\bX^{m,n},\pi_{\bU,\bX})&\le \sum_{i=1}^m H(U_i|\bX^{m,n},\pi_{\bU,\bX})\\
&\le \sum_{i=1}^m \big[\entro(\eps_{U,i})+\eps_{U,i}\log(|\cL|-1)\big]\\
&\le  m\big[\entro(\eps_{U})+\eps_{U}\log(|\cL|-1)\big]\, ,
\end{align*}
where the last step follows by Jensen's inequality.
The claim \eqref{eq:Fano} follows by substituting this bound in Eq.~\eqref{eq:InfoClaim}
and using a similar bound for $H(\bV^n|\bX^{m,n},\pi_{\bV,\bX})$.

\subsection{Proof of Lemma \ref{lemma:AEP}}

We begin with a technical fact.
\begin{lemma}\label{lemma:TableConcentration}
Let $\bxi = (\xi_{ij})_{i\le m,j\le n}$, $\bsigma = (\sigma_{i})_{i\le m}$,
$\btau = (\tau_{j})_{\j\le n}$ be collections of mutually independent random variables
taking values in a measurable space $\cZ$.
$x:\cZ^3\to\cZ$, $F:\cZ^{m\times n}\to\reals$. Define $\bx(\bxi,\bsigma,\btau)\in\reals^{m\times n}$
via $\bx(\bxi,\bsigma,\btau)_{ij} = x(\xi_{ij},\sigma_i,\tau_j)$.

Given a vector of independent random variables $\bz$, we let
$\Var_{z_i}(f(\bz)):=\E_{z_i}[(f(\bz)-\E_{z_i}f(\bz))^2]$.
 Define the quantities
\begin{align}
B_* &:=\max_{\substack{\bx,\bx'\in\cZ^{m\times n}\\
d(\bx,\bx')\le 1}}\big|F(\bx)-F(\bx'')\big|\, ,\\
B_1 &:= 
\max_{\substack{\bsigma,\bsigma'\in\cZ^{m}\\
d(\bsigma,\bsigma')\le 1}} \max_{\btau\in \cZ^n}\big|\E_{\bxi}F(\bx(\bxi,\bsigma,\btau))-\E_{\bxi}F(\bx(\bxi,\bsigma',\btau))\big|\, ,
\label{eq:B1Def}\\
B_2&:= \max_{\bsigma\in \cZ^m} \max_{\substack{\btau,\btau'\in\cZ^{n}\\
d(\btau,\btau')\le 1}} 
 \big|\E_{\bxi}F(\bx(\bxi,\bsigma,\btau))-\E_{\bxi}F(\bx(\bxi,\bsigma,\btau'))\big|\, ,\label{eq:B2Def}\\
V_*&:= \sup_{\bxi,\btau,\bsigma}\sum_{i\le m,j\le n}\Var_{\xi_{ij}}\big\{F(\bx(\bxi,\bsigma,\btau))\big\}\, ,
\label{eq:VstarDef}\\
V_1 &:= \sup_{\btau,\bsigma} \sum_{i\le m} \Var_{\sigma_i}\big\{\E_{\bxi}F(\bx(\bxi,\bsigma,\btau)\big\}\, ,
\label{eq:V1Def}\\
V_2 &:= \sup_{\btau,\bsigma} \sum_{j\le n} \Var_{\sigma_j}\big\{\E_{\bxi}F(\bx(\bxi,\bsigma,\btau))\big\}\, .
\label{eq:V2Def}
\end{align}
Then, for any $t\ge 0$,  the following holds with probability at least $1-8e^{-t}$:
\begin{align}
\big|F(\bx(\bxi,\bsigma,\btau))-\E F(\bx(\bxi,\bsigma,\btau))\|\le 
2\max(\sqrt{2V_*t}+\sqrt{2V_1t}+\sqrt{2V_2t};(B_*+B_1+B_2)t)\, .
\end{align}
\end{lemma}
\begin{proof}
Let $\bz\in\cZ^N$ be a vector of independent random variables and $f:\cZ^N\to \reals$. 
Define the martingale $X_k :=\E[f(\bz)|\cF_k]$ (where $\cF_k:=\sigma(z_1,\dots,z_k)$). Then we have
\begin{align}
 \ess\sup|X_{k}-X_{k-1}|&\le B_0:= \sup_{d(\bz,\bz')\le 1}|f(\bz)-f(\bz')|\, ,\\
 \sum_{k=1}^N\E[(X_{k}-X_{k-1})^2|\cF_{k-1}]& = 
  \sum_{k=1}^N\E\big[(\E[f|\bz_{< k}, z_k]- \E_{z'_k}\E[f|\bz_{< k},z'_k])^2\big|\bz_{<k}\big] \\
 & \le V_0:=
 \sup_{\bz\in\cZ^N}\sum_{k=1}^N\Var_{z_k}\big(f(\bz)\big)\, .
 \end{align}
 By Freedman's inequality, with probability at least $1-2e^{-t}$, we have
 \begin{align}
\big| f(\bz)-\E f(\bz) \big|\le\max\Big(\sqrt{2V_0t}:\, \frac{2B_0t}{3}\Big)\, .
\end{align}

Define $E(\bsigma,\btau) := \E_{\bxi}F(\bx(\bxi,\bsigma,\btau))$, 
$L(\btau):=\E_{\bsigma,bxi}F(\bx(\bxi,\bsigma,\btau))$. Applying the above inequality, 
each of the following holds with probability at least $1-2e^{-t}$
\begin{align}
|F(\bx(\bxi,\bsigma,\btau))-E(\bsigma,\btau)|\le \max\Big(\sqrt{2V_*t}:\, \frac{2B_*t}{3}\Big)\, ,\\
|E(\bsigma,\btau)-L(\btau)|\le \max\Big(\sqrt{2V_1t}:\, \frac{2B_1t}{3}\Big)\, ,\\
|L(\btau)-\E F(\bx)|\le \max\Big(\sqrt{2V_2t}:\, \frac{2B_2t}{3}\Big)\, ,
\end{align}
and the claim follows by union bound.
\end{proof}

We next state and prove a more stronger version of Lemma \ref{lemma:AEP}.
\begin{lemma}
For $\bX\in\cX^{m\times n}$, let  $\rP(\bX) = \rP_{Q,\ro,\co;m,n}(\bX)$ the probability 
of table $\bX$ under  the  model $\cT(Q,\ro,\co;m,n)$, i.e.
\begin{align}
\rP(\bX) = \sum_{\bu\in\cL^m}\sum_{\bv\in\cL^n}\prod_{(i,j)\in[m]\times[n]}
Q(X_{ij}|u_i,v_i)\prod_{i\in[m]} 
\ro(u_i)\prod_{j\in[n]} \co(v_j)\, .
\end{align}
Define the following quantities:
\begin{align}
M_*&:= \max_{x,x'\in\cX}\max_{u,v\in\cL}\Big|\log\frac{Q(x|u,v)}{Q(x'|u,v)}\Big|\, ,\\
M_1&:= \max_{\tau,\sigma,\sigma'}\|Q(\,\cdot\,|\sigma,\tau)-Q(\,\cdot\,|\sigma',\tau)\|_{\sTV}
\max_{u,v,x,x''}\Big|\log\frac{Q(x|u,v)}{Q(x'|u,v)}\Big|\, ,\\
M_2&:= \max_{\tau,\tau',\sigma}\|Q(\,\cdot\,|\sigma,\tau)-Q(\,\cdot\,|\sigma,\tau'')\|_{\sTV}
\max_{u,v,x,x''}\Big|\log\frac{Q(x|u,v)}{Q(x'|u,v)}\Big|\, ,\\
s_* &:= \frac{1}{2}\max_{u_0,v_0\in\cL}\sum_{x,x'\in\cX}Q(x|u_0,v_0)Q(x'|u_0,v_0)
\max_{u,v\in\cL}\Big(\log\frac{Q(x|u,v)}{Q(x'|u,v)}\Big)^2\, ,\\
s_1 &:=  \frac{1}{2}\max_{u_0,u_0',v_0\in\cL}\|Q(\,\cdot\,|u_0,v_0)-Q(\,\cdot\,|u'_0,v_0)\|_{\sTV}
\max_{x,x'\in\cL}\max_{u,v\in\cL}\Big(\log\frac{Q(x|u,v)}{Q(x'|u,v)}\Big)^2\, ,\\
s_2 &:=  \frac{1}{2}\max_{u_0,v_0,v_0'\in\cL}\|Q(\,\cdot\,|u_0,v_0)-Q(\,\cdot\,|u_0,v'_0)\|_{\sTV}
\max_{x,x'\in\cL}\max_{u,v\in\cL}\Big(\log\frac{Q(x|u,v)}{Q(x'|u,v)}\Big)^2\, .
\end{align}
Then, for $\bX\sim\cT(Q,\ro,\co;m,n)$ and any 
$t\ge 0$ the following bound holds with probability at least $-2\, e^{-t}$:
\begin{align}
\big|-\log\rP(\bX) - H(\bX)\big|\le  3\max\Big(\sqrt{s_*mnt}+\sqrt{s_1mn^2t}+
\sqrt{s_2m^2nt},M_*+M_1n+M_2m\Big)\, .
\end{align}
\end{lemma}
\begin{proof}
Let $\bsigma = (\sigma_i)_{i\le m}\sim_{iid}r$, $\btau = (\tau_i)_{i\le n}\sim_{iid}c$,
$\bxi = (\xi_{ij})_{i\le m, j\le n}\sim_{iid}\Unif([0,1])$, and $x:[0,1]\times\cL\times\cL\to \cX$
be such that $x(\xi_{ij},\sigma_i,\tau_j)|_{\sigma_i,\tau_j}\sim Q(\,\cdot\,|\sigma_i,\tau_j)$.
We define $F(\bx) = -\log\rP(\bx)$, and will apply Lemma \ref{lemma:TableConcentration}
to this function. Using the notation from that lemma, we claim that $B_*\le M_*$, $B_1\le M_1n$, $B_2\le M_2m$,
and $V_*\le mn s_*$, $V_1\le mn^2 s_1$, $V_2\le m^2n s_2$. 

Note that, if $(x_{ij})$, $(x'_{ij})$ differ only for entry $i,j$,
then
\begin{align}
F(\bx)-F(\bx') = -\log\rE_{\bu,\bv|\bx}\Big\{\frac{Q(x'_{ij}|u_i,v_j)}{Q(x_{ij}|u_i,v_j)}
\Big\}\, ,
\end{align}
where $\rE_{\bu,\bv|\bx}$ denotes expectation with respect to the posterior measure
$\rP(\bu,\bv|\bX=\bx)$. This immediately implies $B_*\le M_*$.

Next consider the constant $B_1$ defined in Eq.~\eqref{eq:B1Def}. Using the exchangeability of 
the $(\xi_{i,\cdot},\sigma_i)$, we get
\begin{align*}
B_1 &= \max_{\btau}\big|\E_{\bxi}\log
\rE_{\bu,\bv|\bx}\Big\{\prod_{j=1}^n\frac{Q(x(\xi_{1,j},\sigma'_1,\tau_j)|u_1,v_j)}{Q(x(\xi_{1,j},\sigma_1,\tau_j)|u_1,v_j)}\Big\}
\big|\\
&\le \max_{\btau}\E_{\bxi}\max_{\bu,\bv}
\Big|\log\Big\{\prod_{j=1}^n\frac{Q(x(\xi_{1,j},\sigma'_1,\tau_j)|u_1,v_j)}{Q(x(\xi_{1,j},\sigma_1,\tau_j)|u_1,v_j)}\Big\}\Big|\\
&\le \max_{\btau}\sum_{j=1}^n\E_{\bxi}\max_{\bu,\bv}
\Big|\log\Big\{\frac{Q(x(\xi_{1,j},\sigma'_1,\tau_j)|u_1,v_j)}{Q(x(\xi_{1,j},\sigma_1,\tau_j)|u_1,v_j)}\Big\}\Big|\\
&\le n\max_{\tau,\sigma,\sigma'}\E_{\xi}\max_{u,v}
\Big|\log\frac{Q(x(\xi,\sigma'_1,\tau_j)|u,v)}{Q(x(\xi,\sigma_1,\tau_j)|u,v)}\Big|\\
&\le n\max_{\tau,\sigma,\sigma'}\|Q(\,\cdot\,|\sigma,\tau)-Q(\,\cdot\,|\sigma',\tau)\|_{\sTV}
\max_{u,v,x,x'}\Big|\log\frac{Q(x|u,v)}{Q(x'|u,v)}\Big|=M_1\, .
\end{align*}
The bound $B_2\le M_2m$ is proved analogously.

Consider now the quantity $V_*$ of Eq.~\eqref{eq:VstarDef}. 
Denote by $\bxi_{(ij)}(t)$ the array obtained by replacing entry $(i,j)$ in $\bxi$ by $t$,
and by $\bx(t)=\bx(\bxi_{(ij)}(t),\bsigma,\btau)$.
Then we have
\begin{align*}
\Var_{\xi_{ij}}(F(\bx)) &= \frac{1}{2} \E_{\xi',\xi''}
\Big\{\big(F(\bx(\bxi_{(ij)}(\xi'),\bsigma,\btau)) - F(\bx(\bxi_{(ij)}(\xi''),\bsigma,\btau))\big)^2\Big\}\\
&  =\frac{1}{2} \E_{\xi',\xi''}
\Big\{\Big(\log \rE_{\bu,\bv|\bx(\xi')}\Big\{\frac{Q(x(\xi'',\sigma_i,\tau_j)|u_i,v_j)}{Q(x(\xi',\sigma_i,\tau_j)|u_i,v_j)}
\Big\}\Big)^2\Big\}\\
&\le \frac{1}{2} \E_{\xi',\xi''}\max_{u,v}\Big(
\log\Big\{\frac{Q(x(\xi'',\sigma_i,\tau_j)|u,v)}{Q(x(\xi',\sigma_i,\tau_j)|u,v)}
\Big\}\Big)^2\\
&=\frac{1}{2} \sum_{x,x'} Q(x|\sigma,\tau)Q(x'|\sigma,\tau)
\max_{u,v}\Big(
\log\Big\{\frac{Q(x|u,v)}{Q(x'|u,v)}
\Big\}\Big)^2\, .
\end{align*}
We then have, as claimed
\begin{align*}
V_*&\le \max_{\bxi,\bsigma,\btau} \sum_{i\le m,j\le n}\Var_{\xi_{ij}}\big\{F(\bx)\big\}\\
&\le mn \max_{\bxi,\bsigma,\btau}\Var_{\xi_{ij}}\big\{F(\bx)\big\} \le mn s_*\, .
\end{align*}

Finally consider the quantity $V_1$ of Eq.~\eqref{eq:V1Def} (the argument is similar
for $V_2$). Denote by $\bsigma_{(i)}(t)$ the vector obtained by replacing entry $i$ in $\bsigma$ by $t$.
Proceeding as above, we have
\begin{align*}
\Var_{\sigma_i}(\E_{\bxi}F(\bx)) &= \frac{1}{2} \E_{\sigma',\sigma''}
\Big\{\big(\E_{\bxi}F(\bx(\bxi,\bsigma_{(i)}(\sigma'),\btau)) - \E_{\bxi}F(\bx(\bxi,\bsigma_{(i)},\btau))\big)^2\Big\}\\
&  =\frac{1}{2} \E_{\sigma',\sigma''}
\Big\{\Big(\E_{\bxi}\log \rE_{\bu,\bv|\bx(\sigma')}\Big\{
\prod_{j=1}^n\frac{Q(x(\xi_{ij},\sigma'',\tau_j)|u_i,v_j)}{Q(x(\xi_{ij},\sigma',\tau_j)|u_i,v_j)}
\Big\}\Big)^2\Big\}\\
&\le \frac{1}{2} \E_{\sigma',\sigma''}
\Big\{\Big(\E_{\bxi}\log \Big\{
\prod_{j=1}^n\max_{u,v}\frac{Q(x(\xi_{ij},\sigma'',\tau_j)|u,v)}{Q(x(\xi_{ij},\sigma',\tau_j)|u,v)}
\Big\}\Big)^2\Big\}\\
&\le  \frac{1}{2} \E_{\sigma',\sigma''}
\Big\{\Big(\sum_{j=1}^nE_{\xi}\log \Big\{
\max_{u,v}\frac{Q(x(\xi,\sigma'',\tau_j)|u,v)}{Q(x(\xi,\sigma',\tau_j)|u,v)}
\Big\}\Big)^2\Big\}\\
&\le  \frac{n^2}{2} \max_{\tau}\E_{\sigma',\sigma''}
\Big\{\Big(\E_{\xi}\log \Big\{
\max_{u,v}\frac{Q(x(\xi,\sigma'',\tau)|u,v)}{Q(x(\xi,\sigma',\tau)|u,v)}
\Big\}\Big)^2\Big\}\\
&\le \frac{n^2}{2} \max_{\tau,\sigma,\sigma'}\|Q(\,\cdot\,|\sigma,\tau)-Q(\,\cdot\,|\sigma',\tau)\|_{\sTV}
\max_{x,x'}\Big(\E_{\xi}\log \Big\{
\max_{u,v}\frac{Q(x'|u,v)}{Q(x|u,v)}
\Big\}\Big)^2 = n^2 s_1\, .
\end{align*}
Therefore
\begin{align*}
V_1 = \max_{\bsigma,\btau}\sum_{i=1}^m\Var_{\sigma_i}(\E_{\bxi}F(\bx))\le mn^2 s_1\, .
\end{align*}
This finishes the proof.
\end{proof}

\section{Proofs for finite state encoders}
\label{sec:FiniteState}

Recall from Section \ref{sec:Failure} that a finite-state encoder is defined by 
a triple $(\Sigma,f,g)$.
Formally we can define the action of $f$, $g$ on $\bX^n\in\cX^n$ recursively via
(recall that $\oplus$ denotes concatenation)
\begin{align}
f_{m+1}(\bX^{m+1},s_0) &= f_{m}(\bX^{m},s_0) \oplus f(X_{m+1},g(\bX^{m},s_0))\, ,\\
g_{m+1}(\bX^{m+1},s_0) &= g(X_{m+1},g(\bX^{m},s_0))\, ,
\end{align}
and the encoder is thus given by $E(\bX^n) = f_n(\bX^{n},s_{\sinit})$.

We say that the state space $\Sigma$ is non-degenerate if, for each $s_1\in \Sigma$
there exists $m$, $\bX^m\in\cX^m$ such that $g_{m}(\bX^{m},s_{\sinit})=s_1$.
Notice that if state space is degenerate, we could always remove one or more symbols from $\Sigma$
without changing the encoder, and making the state-space non-degenerate.
For this reason, we will hereafter assume non-degeneracy without mentioning it.

We say that the FS encoder is information lossless (IL) if for any $n\in\naturals$,
$\bX^n\mapsto f_n(\bX^{n},s_{\sinit})$ is injective.
\begin{remark}
An information-lossless encoder satisfies a stronger condition:
for any $m\in\naturals$ and any $s_*\in \Sigma$, the map $\bX^m\mapsto f_m(\bX^{m},s_*)$
is injective.

Indeed, assume this were not the case. Then there would exist two distinct inputs
 $\bX^m$, $\tbX_1^m\in\cX^m$ and a state $s_*\in \Sigma$ such that 
$f_m(\bX^{m},s_*)=f_m(\tbX^{m},s_*)$. By non-degeneracy, there exists 
$a_1^\ell\in\cX^\ell$ such that $s_* = g_{\ell}(a_1^{\ell},s_{\init})$,
Defining $n=\ell+m$, $\bY^n = a_1^{\ell}\oplus \bX^m$,
$\tbY^n = a_1^{\ell}\oplus \tbX^m$, it is not hard to check that these inputs
are distinct but $f_n(\bY^{n},s_{\init})=f_n(\tbY_1^{n},s_{\init})$.
\end{remark}

\begin{proposition}\label{propo:FiniteState}
Define the compression rate on input $x_1^n$ as $\Rate(\bX^n) = \len(f_n(\bX^n,s_{\init}))/(n\log_2|\cX|)$.
Then for any $\ell\ge 1$, the following holds (where $n':=n-2\ell$ and we recall
that $M:=|\Sigma|$):
\begin{align}
\Rate(\bX^n)&\ge \frac{n-2\ell}{n\ell\log_2|\cX|}H(\hp^{\ell}_{\bX_1^{n'}})
-\frac{1}{\ell\log_2|\cX|}\big(\log_2 (|\Sigma|\ell)+\log_2\log_2|\cX|\big)\, .\label{eq:RateLB}
\end{align}
\end{proposition}
\begin{proof}
We will denote by $L(\bX^m;s_*)$ the length of the encoding of $\bX^m$ 
when starting in state $s_*$:
\begin{align}
L(\bX^m;s_*) := \len(f_n(\bX^m,s_*))\,.
\end{align}
We then have, for any $b\in\{0,\dots,\ell-1\}$, and setting by convention $s_0=s_{\init}$,
we get
\begin{align}
\Rate(\bX^n) &\ge \frac{1}{n\log_2|\cX|} \sum_{k=0}^{\lfloor n/\ell\rfloor -2}L(\bX_{k\ell+b+1}^{(k+1)\ell+b};s_{k\ell+b})\, .
\end{align}
By averaging over $b$, and introducing the shorthand $n':=n-2\ell$, we get
\begin{align}
\Rate(\bX^n)&\ge \frac{1}{n\ell\log_2|\cX|} \sum_{m=1}^{(\lfloor n/\ell\rfloor -1)\ell}
L(\bX_{m}^{m+\ell-1};s_{m-1})\\
&\ge \frac{n-2\ell}{n\ell\log_2|\cX|} \sum_{s\in \Sigma}\sum_{u_1^\ell\in\cX^{\ell}}
\hp^{\ell}_{\bX_1^{n'}}(u_1^{\ell},s)\, L(u_{1}^{\ell};s)\\
&\stackrel{(a)}{\ge}  \frac{n-2\ell}{n\ell\log_2|\cX|} \sum_{s\in \Sigma} 
\Big\{ \hp^{\ell}_{\bX_1^{n'}}(s)\, H(\hp^{\ell}_{\bX_1^{n'}}(\,\cdot\,|s))
-\log_2\log_2(|\cX|^{\ell})\Big\}\, ,\label{eq:FirstRB}
\end{align}
where $(a)$ holds by Lemma \ref{lemma:SourceCodingConverse}.
By the chain rule of entropy  (recalling
that $M:=|\Sigma|$), we have:
\begin{align*}
  \sum_{s\in \Sigma} \hp^{\ell}_{\bX_1^{n'}}(s)\, H(\hp_{\bX_1^{n'}}^{\ell}(\,\cdot\,|s))
 &= H(\bX_{1}^{\ell}|S)= H(\bX_1^{\ell}) +H(S|\bX_1^{\ell})-H(S)\\
 &\ge H(\bX_1^{\ell})  - \log_2 M = H(\hp^{\ell}_{\bX_1^{n'}})-\log_2 M\, .
 \end{align*}
The claim \eqref{eq:RateLB} follows by using the last inequality in Eq.~\eqref{eq:FirstRB}.
\end{proof}

\begin{theorem}
Let $\bX^{m,n}\sim\cT(Q,\ro,\co;m,n)$ and $(\Sigma,f,g)$ be an information lossless
 finite state  encoder.
With an abuse of notation, denote $f_{mn}(\bX^{m\times n},s_{\init})\in\{0,1\}^*$ the binary sequence
obtained by applying the finite state encoder to the vector
$\vec(\bX^{m\times n})\in\cX^{mn}$ obtained by scanning $\bX^{m\times n}$ in row-first order.
Define the compression rate by
\begin{align}
\Rate(\bX^{m,n})&:=\frac{\len(f_{mn}(\bX^{m\times n},s_{\init}))}{mn\log_2|\cX|}\,.
\end{align}

Assuming $m>10$, $|\Sigma|\ge |\cX|$, and $\log_2|\Sigma|\le n\log_2|\cX|/9$,
the expected compression rate is lower bounded as follows 
\begin{align}
\E\, \Rate(\bX^{m,n})&\ge \frac{H(X|U)}{\log_2|\cX|}-
10 \sqrt{\frac{\log|\Sigma|}{n\log|\cX|}} \cdot\log(n\log|\Sigma|)\, .
\end{align}
\end{theorem}
\begin{proof}
We let $N:=mn$, $N'=mn-2\ell$ 
where we $\ell\le n/3$ will be selected later. We write $\bX^N:=\vec(\bX^{m,n})$ for
the vectorization $\bX^{m,n}$, $\bX^{N'}$ for the vector comprising its first $N'$ entries.
Recall the definition of empirical distribution. For any fixed $\bw\in\cX^{\ell}$
\begin{align*}
\hp^{\ell}_{\bX^{N'}}(\bw):=\frac{1}{N'-\ell+1}\sum_{i=1}^{N'-\ell+1}
\bfone_{\bX_i^{i+\ell-1}=\bw}\, .
\end{align*}
Let $S:=\{i\in [N'-\ell+1]: \, [i,i+\ell-2]\cap n\naturals =\emptyset\}$. In
words, these are the subset of blocks of length $\ell$ that do not cross the end of
a line in the table. Since for each line break there are at most $\ell-1$
such blocks, we have $|S|\ge N'-\ell+1-(m-1)(\ell-1)$.
We will consider the following modified empirical distribution
\begin{align*}
\hhp^{\ell}_{\bX^{N'}}(\bw):=\frac{1}{|S|}\sum_{i\in S}
\bfone_{\bX_i^{i+\ell-1}=\bw}\, .
\end{align*}
Then by construction
\begin{align*}
\hp^{\ell}_{\bX^{N'}}(\bw) & =  (1-\eta_{\ell})\hhp^{\ell}_{\bX^{N'}}(\bw)+
\eta_{\ell}q^{\ell}_{\bX^{N'}}(\bw)\, ,\\
\eta_{\ell} & :=1- \frac{|S|}{N'-\ell+1} = \frac{(m-1)(\ell-1)}{N'-\ell+1}\, ,
\end{align*}
where $q^{\ell}_{\bX^{N'}}$ is the empirical distribution of blocks that do cross the line.
By concavity of the entropy, we have
\begin{align}\label{eq:EdgeEntropy}
H(\hp^{\ell}_{\bX^{N'}})\ge (1-\eta_{\ell}) H(\hhp^{\ell}_{\bX^{N'}})+\eta_{\ell}
H(q^{\ell}_{\bX^{N'}}) \ge (1-\eta_{\ell}) H(\hhp^{\ell}_{\bX^{N'}})\, .
\end{align}
Further, since $\ell\le n/3$, 
\begin{align}
\eta_{\ell}&=
\frac{(m-1)(\ell-1)}{mn-3\ell+1}\nonumber\\
&\le \frac{(m-1)\ell}{mn-3\ell}\le  \frac{(m-1)\ell}{(m-1)n}
\le \frac{\ell}{n}\, .\label{eq:EdgeEffect}
\end{align}
Now let the row latents $\bu := (u_i)_{i\le m}$ be fixed, and denote by $\hr^S_{\bu}$ their 
weighted empirical distribution, defined as follows:
\begin{align*}
\hr^{S}_{\bu}(s) := \sum_{i=1}^m\frac{|S\cap[(i-1)n+1,in]|}{|S|}\bfone_{u_i=s}\, .
\end{align*}
In words, $\hr^{S}_{\bu}$ is the empirical distribution of the latents $(u_i)_{i\le m}$
where row $i$ is weighted by its contribution to $S$. Note that all the weights are equal 
to $(n-2(\ell-1))/|S|$ except, potentially, for the last one.

 We have 
\begin{align*}
p^{\ell}_*(\bw):=\E[\hhp^{\ell}_{\bX^{N'}}(\bw)] =
 \sum_{u\in\cL}\hr^{S}_{\bu}(u)\prod_{i=1}^{\ell}Q_{x|u}(w_i|u)\, ,\;\;\;\;
 Q_{x|u}(w|u):=\sum_{v\in\cL} Q(w|u,v) \, \co(v)\, .
\end{align*}
Using Eq.~\eqref{eq:EdgeEntropy}, \eqref{eq:EdgeEffect} and concavity of the entropy, we get
\begin{align}
\E \big[H(\hp^{\ell}_{\bX^{N'}})|\bu\big]\ge \Big(1-\frac{\ell}{n}\Big) H(p^{\ell}_*)\, .
\end{align}

By Proposition \ref{propo:FiniteState}, we get
\begin{align*}
\E \,\big[\Rate(\bX^{m,n})|\bu\big]\ge & \frac{mn-2\ell}{mn\ell\log_2|\cX|}
\Big(1-\frac{\ell}{n}\Big)H(p^{\ell}_*)
-\frac{1}{\ell\log_2|\cX|}\big(\log_2 (|\Sigma|\ell)+\log_2\log_2|\cX|\big) \\
\ge &\frac{1}{\ell\log_2|\cX|}H(p^{\ell}_*)-\frac{2\ell}{n}
-\frac{1}{\ell\log_2|\cX|}\big(\log_2 (|\Sigma|\ell)+\log_2\log_2|\cX|\big)\, ,
\end{align*}
where in the last inequality we used the fact that $H(p^{\ell}_*)\le \ell\log_2|\cX|$.
We choose
\begin{align}
\ell = \sqrt{\frac{n\log_2|\Sigma|}{\log_2|\cX|}}\le \frac{n}{3}\, ,
\end{align}
Substituting and simplifying, we get
\begin{align}
\E \big[\Rate(\bX^{m,n})|\bu\big]&\ge \frac{H(p^{\ell}_*)}{\ell\log_2|\cX|}-
\frac{10}{\sqrt{n}}\cdot \sqrt{\frac{\log|\Sigma|}{\log|\cX|}} \cdot\log(n\log|\Sigma|)\, .
\end{align}
Finally, letting $(W_1,\dots, W_{\ell},U)\in\cX^{\ell}\times \cL$ be random
variables with joint distribution $\hr^{S}_{\bu}(u)\prod_{i=1}^{\ell}Q_{x|u}(w_i|u)$.
Then
\begin{align}
H(p^{\ell}_*) &\ge \sum_{u\in \cL} \hr^{S}_{\bu}(u)H\big(Q^{\otimes \ell}_{x|u}(\,\cdot\,|u)\big)\\
& \ge \ell \sum_{u\in \cL} \hr^{S}_{\bu}(u) H(X|U=u)\, ,
\end{align}
and therefore $\E H(p^{\ell}_*)\ge H(X|U)$, finishing the proof.
\end{proof}

\section{Proofs for Lempel-Ziv coding}
\label{sec:ProofLZ}

The pseudocode of the Lempel-Ziv algorithm that we will analyze is given here.
For ease of presentation, we identify $\cX$ with a set of integers.

\begin{algorithm}[tb]
\caption{Lempel-Ziv}
\label{alg:lempelziv}
\begin{algorithmic}
\STATE {\bfseries Input:} {Data $\bX^N\in\cX^{N} = \{0,\cdots,|\cX|-1\}^N$}
\STATE {\bfseries Output:} {Binary string $\bZ\in \{0,1\}^*$}
\FOR{$k=1$ {\bfseries to} $N$}
\IF{$\exists j<k: X_j=X_k$}
\STATE $L_{k}\leftarrow \max\{\ell\ge 1:\; \exists j\in\{1,\dots,k-1\} \mbox{ s.t. } \bX_{j}^{j+\ell-1}
=\bX_{k}^{k+\ell-1}\}$\\
$T_{k}\leftarrow \max\{ j\in\{1,\dots,k-1\} \mbox{ s.t. } \bX_{j}^{j+L_k-1}
=\bX_{k}^{k+L_k-1}\}$\;
\ELSE
\STATE $L_k\leftarrow 1$\\
$T_{k}\leftarrow (-X_k)$\;
 
\ENDIF
\STATE $\bZ\leftarrow \bZ \concat \plain(T_k) \concat \elias(L_k)$\;
\STATE $k\leftarrow k+L_k$
\IF{$\len(\bZ) \le \len(\plain(\bX^N)) $}
\STATE {\bfseries return} $\bZ$
\ELSE
\STATE {\bfseries return} $\plain(\bX^N)$
\ENDIF

\ENDFOR

\end{algorithmic}
\end{algorithm}

Note that if a simbol $X_k$ never appeared in the past, 
we point to $T_k=-X_k$ and set $L_k=1$. This is essentially equivalent to prepending 
a sequence of distinct $|\cX|$ symbols to $\bX^N$. 

It is useful to define for each $k\le N$,  
\begin{align}
L_{k}(\bX^N) &:=\max\big\{\ell\ge 1:\; \exists j\in\{1,\dots,k-1\} \mbox{ s.t. } \bX_{j}^{j+\ell-1}
=\bX_{k}^{k+\ell-1}\big\}\, ,\label{eq:LK-def}\\
T_{k}(\bX^N) &:=\min\big\{ j\in\{1,\dots,k-1\} \mbox{ s.t. } \bX_{j}^{j+L_k-1}
=\bX_{k}^{k+L_k-1}\big\}\, .\label{eq:TK-def}
\end{align}

\subsection{Proof of Theorem \ref{thm:LempelZiv}}

\begin{lemma}\label{lemma:CoarseBoundL}
Under Assumption \ref{ass:NonDet}, there exists a constant $C$ such that
the following holds with probability at least $1-N^{-10}$:
\begin{align}
\max_{k\le N}L_k(\bX^N)\le C\, \log N\, .
\end{align}
\end{lemma}
\begin{proof}
We begin by considering a slightly different setting, and will then show that 
our question reduces to this setting. Let $(Z_i)_{i\ge 1}$ be independent random variables with 
$Z_i\sim q_i$ a probability distribution over $\cX$. Further assume $\max_{x\in\cX}q_i(x)\le 1-c$
for all $i\ge 1$. Then we claim that, for any $t,\ell\ge 1$, we have
\begin{align}
\prob\big(Z_1^{\ell} = Z_{t+1}^{t+\ell}\big) \le (1-c)^{\ell}\, .\label{eq:ClaimZ}
\end{align}
Indeed, condition on the event $Z_1^{t}=x_{1}^t$ for some $x_1,\dots,x_t\in\cX$.
Then the event $Z_1^{\ell} = Z_{t+1}^{t+\ell}$ implies that, for $i\in \{t+1,\dots,t+\ell\}$,
$Z_i = x_{\pi(i)}$ where $\pi(i) = i \mod t$, $\pi(i)\in [1,t]$. Then 
\begin{align*}
\prob\big(Z_1^{\ell} = Z_{t+1}^{t+\ell}\big)&\le \max_{x_1^t\in\cX^t}
\prob\big(Z_1^{\ell} = Z_{t+1}^{t+\ell}| Z_{1}^t=x_1^t\big)\\
&\le \max_{x_1^t\in\cX^t}
\prob\big(Z_i = x_{\pi(i)}\forall i\in\{t+1,\dots,t+\ell\}| Z_{1}^t=x_1^t\big)\\
&\le \max_{x_1^t\in\cX^t}\prod_{i=t+1}^{t+\ell} \prob\big(Z_i = x_{\pi(i)}\big)\le
(1-c)^{\ell}\, .
\end{align*}
This proves claim \eqref{eq:ClaimZ}.

Let us now reconsider our original setting:
\begin{align*}
\prob\big(\max_{k\le N}L_k(\bX^N)\ge \ell\big) &=  \prob\big(\exists i<j\le N:\;
X_{i}^{i+\ell-1}=X_{j}^{j+\ell-1}\big) \\
&\le N^2 \max_{i<j\le N}\prob\big(X_{i}^{i+\ell-1}=X_{j}^{j+\ell-1}\big)\\
&\le N^2 \max_{\bu^m\in\cL^m,\bv^n\in \cL^n}\max_{i<j\le N}
\prob\big(X_{i}^{i+\ell-1}=X_{j}^{j+\ell-1}\big|\bu^m,\bv^n\big)\\
&\le N^2 (1-c)^{\ell}\, ,
\end{align*}
where the last inequality follows from claim \eqref{eq:ClaimZ}, since the $(X_i)_{i\le N}$ are
conditionally independent given the latents $\bu^m, \bv^n$, with probability mass function
upper bounded by $1-c$. The thesis follows by taking $\ell = 12\log N/\log(1/(1-c))$.
\end{proof}

For $i\in [m]$, $j\in [n]$, we define $\<ij\>:=(i-1)n+j$. In words, $k=\<ij\>$
is the of entry at row $i$ column $j$ when the table $\bX^{m,n}$ is scanned in row first order.
For $\ell\ge 1$, define the events
\begin{align} 
\cE_{i,j}(\ell) &:= \Big\{ \exists i'\in[m],j'\in[n]\; :\; \<i'j'\>< \<ij\>, \;
|j'-j|\ge \ell ,\; \bX_{\<i'j'\>}^{\<i'j'\>+\ell-1}=\bX_{\<ij\>}^{\<ij\>+\ell-1}\}\, ,\\
\cF_{i,j}(\ell) &:= \Big\{ \exists i'\in[m],j'\in[n]\; :\; \<i'j'\>< \<ij\>, \;
|j'-j|<\ell ,\; \bX_{\<i'j'\>}^{\<i'j'\>+\ell-1}=\bX_{\<ij\>}^{\<ij\>+\ell-1}\}\, .
\end{align}
Then we have
\begin{align}
\prob\big(L_{\<ij\>}(\bX^N)\ge \ell\big)\le \prob\big(\cE_{i,j}(\ell)\big)+
\prob\big(\cF_{i,j}(\ell)\big)\, .
\end{align}
The next two lemmas control the probabilities of these events.
%
\begin{lemma}\label{lemma:NoCol}
Let $\ell(\delta,u) := \lceil (1+\delta)(\log N)/H(X|U=u)\rceil$,
 $n'=n-\max_{u\in\cL}\ell(\delta,u)$, and $m_0= m^{1-o_n(1)}$.
Under Assumption \ref{ass:NonDet}, for any $\delta>0$, there exist constants 
$C,\eps>0$ independent of $\bu\in\cL^m$, such that
the following hold
\begin{align}
&\max_{i\le m,j\le n'}\prob\big(\cE_{i,j}(\ell(\delta,u_i))\big)\le C\, N^{-\eps}\, ,\label{eq:UB-ell}\\
&\min_{m_0\le i\le m,j\le n'}\prob\big(\cE_{i,j}(\ell(-\delta,u_i))\big)\ge 1-C\, N^{-\eps}\, .\label{eq:LB-ell}
\end{align}
\end{lemma}

\begin{lemma}\label{lemma:Col} 
Let $\ell_c(\delta,u) := \lceil (1+\delta)(\log m)/H(X|U=u,V)\rceil$,
$n'_c=n-\max_{u\in\cL}\ell_c(\delta,u)$, and $m_0= m^{1-o_n(1)}$.
Under Assumption \ref{ass:NonDet}, for any $\delta>0$, there exist constants $C,\eps>0$ 
independent of $\bu\in\cL^m$,  such that
the following hold
\begin{align}
&\max_{i\le m,j\le n_c'}\prob\big(\cF_{i,j}(\ell_c(\delta,u_i))\big)\le C\, m^{-\eps}\, ,\label{eq:UB-ell-c}\\
&\min_{m_0\le i\le m,j\le n_c'}\prob\big(\cF_{i,j}(\ell_c(-\delta,u_i))\big)\ge 1-C\, m^{-\eps}\, .\label{eq:LB-ell-c}
\end{align}
\end{lemma}

We are now in position to prove Theorem  \ref{thm:LempelZiv}.
\begin{proof}[Proof of Theorem \ref{thm:LempelZiv}]
We denote by $(k(1),\dots, k(M))$ the values taken by $k$ 
in the while loop of the Lempel-Ziv pseudocode. In particular 
\begin{align}
k(1) & = 1\, ,\\
k(\ell+1) & = k(\ell)+L_{k(\ell)}(\bX^N)\, ,\\
k(M)  & = N\, .
\end{align}
Therefore the total length of the code is 
\begin{align}
\len(\LZ(\bX^{m,n})) = M\lceil\log_2 (N+|\cX|)\rceil +\sum_{\ell=1}^M \len(\elias(L_{k(\ell)}))
\end{align}
By Lemma \ref{lemma:CoarseBoundL} (and recalling that $\len(\elias(L))\le 2\log_2 L+1$)
we have, with high probability,
 $\max_{\ell\le m}\len(\elias(L_{k(\ell)}))\le 2\log_2(C\log N)$. 
 Letting $\cG$ denote the `good' event that this bound holds, we have, on $\cG$
\begin{align}
 M\log_2 N  \le \len(\LZ(\bX^{m,n})) \le  M\lceil\log_2 (N+|\cX|)\rceil +2M \log_2(C\log N)
\end{align}
Since $|\cX|$ is a constant, this means that for any $\eta>0$, there exists 
$N_0(\eta)$ such that, for all $N\ge N_0(\eta)$, with probability at least $1-\eta$:
\begin{align}
 M\cdot \bfone_{\cG}\log_2 N  \le \len(\LZ(\bX^{m,n})) \le  (1+\eta)M\cdot\bfone_{\cG}\log_2 N + 
 N\cdot \bfone_{\cG^c}  \log_2 |\cX|\, ,
\end{align}
where on the right $\len(\LZ\bX^{m,n}))\le N  \log_2 |\cX|$ by construction.
We thus have
\begin{align}
 \E\big\{ M\cdot \bfone_{\cG}\big\} \frac{\log_2 N}{N\log_2|\cX|}  \le \E\,\Rate_{\LZ}(\bX^{m,n})
  \le  (1+\eta)\E\big\{ M\cdot \bfone_{\cG}\big\} 
  \frac{\log_2 N}{N\log_2|\cX|} +\eta\,,
\end{align}
that is
\begin{align}
\lim\inf_{m,n\to\infty}\E\,\Rate_{\LZ}(\bX^{m,n})& \ge\lim\inf_{m,n\to\infty}
\E\big\{ M\cdot \bfone_{\cG}\big\} \cdot \frac{\log_2 N}{N\log_2|\cX|}\, ,\label{eq:LB_LZ}\\
\lim\sup_{m,n\to\infty}\E\,\Rate_{\LZ}(\bX^{m,n})& \le\lim\sup_{m,n\to\infty}
\E\big\{ M\cdot \bfone_{\cG}\big\}  \cdot \frac{\log_2 N}{N\log_2|\cX|}\, .\label{eq:UB_LZ}
\end{align}
We are therefore left with the task of bounding $\E\big\{ M\cdot \bfone_{\cG}\big\}$

We begin by the lower bound. Define the set of `bad indices' $B(\bX^{m,n},\delta)\subseteq[m]\times[n]$, 
\begin{align}
B(\bX^{m,n},\delta):=\Big\{(i,j)\in [m]\times[n]:\; \cE_{i,j}(\ell(\delta,u_i)) \mbox{ or }
\cF_{i,j}(\ell_c(\delta,u_i)) 
\Big\}
\end{align} 
We will drop the arguments $\bX^{m,n},\delta$ for economy of notation, and write
$B:=B(\bX^{m,n},\delta)$. We further define
\begin{align}
S(u)=S(u;\bX^{m,n}):=\{(i,j)\in [m]\times [n]:\; u_i=u
\mbox{ and }
\exists \ell\le M:\; \<ij\>= k(\ell)\}\, .
\end{align}
In words, $S(u)$ is the set of positions $(i,j)$  of the table $\bX^{m,n}$
where words in the LZ parsing begin.

 We also write $N(u)= n\cdot|\{i\in[m]:\ \, u_i=u\}|$ 
for the total number of rows in $\bX^{m,n}$ with row latent equal to $u_i$
and $L^-_i$ for the length of the first segment in row $i$ initiated in row $i-1$:
\begin{align*}
N(u) &\le \sum_{(i,j)\in S(u)}L_{\<ij\>}+\sum_{i\le m: u_i=u} L^-_i\\
&\le \sum_{(i,j)\in S(u)\cap B^c}L_{\<ij\>} +  \sum_{(i,j)\in B}L_{\<ij\>}
+\sum_{i\le m: u_i=u} L^-_i\\
&\le \sum_{(i,j)\in S(u)}\ell(u;\delta)\vee \ell_c(u;\delta)+(|B|+m)\cdot C\log N\\
& \le |S(u)|\ell(u;\delta)\vee \ell_c(u;\delta)+(|B|+m)\cdot C\log N\, ,
\end{align*}
where the last inequality holds on event $\cG$. By taking expectation on this event,
we get
\begin{align*}
\E\{N(u)\cdot\bfone_{\cG}\} \le \E\{|S(u)|\cdot\bfone_{\cG}\}\}
\cdot \ell(u;\delta)\vee \ell_c(u;\delta)+ (\E|B|+m)\cdot C\log N\, .
\end{align*}
By Lemmas \ref{lemma:NoCol} and \ref{lemma:NoCol}, 
\begin{align*}
\E(|B|)&\le m_0n+\sum_{m_0\le i\le m,j\le n'}
\prob\Big(\cE_{i,j}(\ell(\delta,u_i))\cup
\cF_{i,j}(\ell_c(\delta,u_i))\Big) + Cm\log N\\
&\le m_0n+Cm^{1-\eps}n+ Cm\log N\\
& \le \frac{CN}{(\log N)^2}\, .
\end{align*}
$\E(|B|)\le Cm^{1-\eps}n + Cm\log n$.
Further $\E\{N(u)\} = N\ro(u)$ and $\E\{N(u)\cdot\bfone_{\cG}\}\ge \E\{N(u)\}-N\,\prob(\cG^c)$,
whence
\begin{align}
\lim\inf_{m,n\to\infty} \frac{1}{N}\E\{|S(u)|\cdot\bfone_{\cG}\}\}
\cdot \ell(u;\delta)\vee \ell_c(u;\delta)\ge \ro(u)\, .
\end{align}
Recalling the definition of $\ell(u;\delta)$, $\ell_c(u;\delta)$ and the fact
that $\delta$ is arbitrary,n the last inequality yields
\begin{align}
\lim\inf_{m,n\to\infty} \E\{|S(u)|\cdot\bfone_{\cG}\}\frac{\log_2 N}{N} 
\ge \ro(u) \Big[H(X|U=u) \wedge \Big(\frac{1+\alpha}{\alpha}\Big)H(X|U=u,V)\Big] \, .
\end{align}
Summing over $u$, noting that $\sum_{u\in\cL}|S(u)|=M$, and substituting in 
Eq.~\eqref{eq:LB_LZ} yields the  lower bound on the rate in Eq.~\eqref{eq:Asymp-LZ}.

Finally, the upper bound is proved by a similar strategy as for the lower bound.  
Define the set of `bad indices' $B_- = B_-(\bX^{m,n},\delta)\subseteq[m]\times[n]$, 
\begin{align}
B_-(\bX^{m,n},\delta):=\Big\{(i,j)\in [m]\times[n]:\; \cE^c_{i,j}(\ell(-\delta,u_i)) \mbox{ or }
\cF^c_{i,j}(\ell_c(-\delta,u_i)) 
\Big\}
\end{align} 
We also denote by $L^+_i$ the length of the last segment in row $i$.
We then have
\begin{align*}
N(u) &\ge \sum_{(i,j)\in S(u)}L_{\<ij\>}-\sum_{i\le m: u_i=u} L^+_i\\
&\ge \sum_{(i,j)\in S(u)\cap B_-^c}L_{\<ij\>} -\sum_{i\le m: u_i=u} L^+_i\\\
&\ge \sum_{(i,j)\in S(u)\cap B_-^c}\ell(u;-\delta)\vee \ell_c(u;-\delta)-\sum_{i\le m: u_i=u} L^+_i\\\
& \ge |S(u)|\ell(u;\delta)\vee \ell_c(u;\delta)-(|B_-|+m)\cdot C\log N\, ,
\end{align*}
where the last inequality holds on event $\cG$. By taking expectation on this event,
we get
\begin{align*}
\E\{N(u)\cdot\bfone_{\cG}\} \ge \E\{|S(u)|\cdot\bfone_{\cG}\}\}
\cdot \ell(-\delta,u)\vee \ell_c(-\delta,u)- (\E|B_-|+m)\cdot C\log N\, .
\end{align*}
By Lemmas \ref{lemma:NoCol} and \ref{lemma:NoCol}, 
\begin{align*}
\E(|B_-|)&\le m_0n+\sum_{m_0\le i\le m,j\le n'}
\prob\Big(\cE_{i,j}^c(\ell(-\delta,u_i))\cup
\cF^c_{i,j}(\ell_c(-\delta,u_i))\Big) + Cm\log N\\
&\le m_0n+Cm^{1-\eps}n+ Cm\log N\\
& \le \frac{CN}{(\log N)^2}\, .
\end{align*}
The proof is completed exactly as for the lower bound.
\end{proof}

\subsection{Proof of Lemma \ref{lemma:NoCol}}

We will use the following standard lemmas.
\begin{lemma}\label{lemma:QuadraticBound}
Let $X$ be a centered random variable with $\prob(X\le x_0) = 1$, $x_0>0$. Then, letting 
$c(x_0) = (e^{x_0}-1-x_0)/x_0^2$, we have
\begin{align}
\E(e^X)\le 1+c(x_0)\E(X^2)\, .
\end{align}
\end{lemma}
\begin{proof}
This simply follows from $\exp(x)\le 1+x+c(x_0)x^2$ for $x\le x_0$.
\end{proof}

\begin{lemma}\label{lemma:HittingTime}
Let  $(p_i)_{i\ge 1}$, $(q_i)_{i\ge 1}$, be probability distributions on 
$\cX$, with 
$\sup_{i\ge 1}\max_{x\in\cX}p_i(x)\le 1-c$, and 
$\sup_{i\ge 1}\sum_{x\in\cX} p_i(x)(\log p_i(x))^2\le C$ for constants $c, C$.

Let $(X_i)_{i\le \ell}$ be independent random variables with $X_i\sim p_i$,
and set $\bX=(X_1,\dots, X_{\ell})$.
 Let $\bY(j)\in\cX^{\ell}$, $j\ge 1$ be a sequence of 
i.i.d. random vectors,
with $(Y_i(j))_{i\le \ell}$ independent and $Y_{i}(j)\sim q_i$.  
Finally, let $T:= \min\{t\ge 1: \; \bY(t) = \bX\}$.

Then, for any $\eps>0$, there exists $\delta=\delta(\eps,c,C)>0$ such that
(letting $\oH(p) := \ell^{-1}\sum_{i=1}^\ell H(p_i)$)
\begin{align}
\prob(T\le e^{\ell [\oH(p)-\eps]})\le e^{-\delta\ell}\, .
\end{align}
Further, the same bound holds (with a different $\delta(\eps,c,C)$)
$(\bY(j))_{j\ge 1}$ are independent not identically distributed, if
there exist a finite set $(q^{a}_i)_{i\ge 1,a\in[K]}$, $K\le \ell^{C_0}$, and a map $b:\naturals\to[K]$
such that $\bY(j)\sim q_1^{b(j)}\otimes \cdots\otimes q_\ell^{b(j)}$.
\end{lemma}
\begin{proof}
We denote by $\bY$ a vector distributed as $\bY(i)$. 
Conditional on $\bX=\bx$, $T$ is a geometric random variables with mean 
$1/(1-\prob(\bY=\bx))$. Hence, for $t_{\ell}(\eps):=e^{\ell [\oH(p)-\eps]}$,
\begin{align*}
\prob(T\le t_{\ell}(\eps)|\bX=\bx)&=1-(1-\prob(\bY=\bx))^{t_\ell(\eps)}\\
& \le t_{\ell(\eps)}\prob(\bY=\bx)\, .
\end{align*}
Hence
\begin{align}
\prob\big(T\le t_{\ell}(\eps)\big)&\le e^{-\ell\eps/2} + 
\prob\Big(\prob(\bY=\bX|\bX)\ge t_{\ell}(\eps/2)^{-1}\Big)\\
&= e^{-\ell\eps/2}+ P_{\ell}(\eps/2)\, ,\\
P_{\ell}(u)&:=\prob\Big(\sum_{i=1}^{\ell}\log \frac{1}{q_i(X_i)}< \sum_{i=1}^{\ell}H(p_i)-\ell\, u\Big)\, .
\end{align}
By Chernoff bound, for any $\lambda\ge 0$, $P_{\ell}(u)\le \exp\{-\ell\phi(\lambda,u)\}$, where
\begin{align}
\phi(\lambda,u):=\lambda u-\frac{1}{\ell}\sum_{i=1}^{\ell}\big[\lambda H(p_i)
+\log \E\big[q_i(X_i)^{\lambda}\big]\, .
\end{align}
By H\"older inequality, for 
$\lambda\in[0,1]$ we have $\E\big[q_i(X_i)^{\lambda}\big]\le (\sum_x p(x)^{\beta})^{1/\beta}$
where $\beta =1/(1-\lambda)$. Therefore
\begin{align*}
\psi(\lambda;p) &:= \lambda H(p)+(1-\lambda) \log\Big(\sum_{x\in\cX}p(x)^{1/(1-\lambda)}\Big) \\
& = (1-\lambda) \log \E_{X\sim p}\exp\Big(\frac{\lambda}{1-\lambda}\big(\log p(X)+H(p)\big)\Big)\, .
\end{align*}
Consider the random variable $Z_i := \frac{\lambda}{1-\lambda}\big(\log p_i(X_i)+H(p)\big)$
where $X_i\sim p_i$. Under the assumptions of the lemma, for $\lambda\in [0,1/2]$
we have $\E(Z_i)=0$ and 
\begin{align}
Z_i & \le \log(1-c) +H(p) \le \log\big[|\cX|(1-c)\big]\, ,\\
\E[Z_i^2] & \le \left(\frac{\lambda}{1-\lambda}\right)^2 
\sum_{x\in\cX}p_i(x) (\log p_i(x))^2\le 4C\lambda^2\, ,
\end{align}
Using Lemma \ref{lemma:QuadraticBound}, we get
\begin{align}
\psi(\lambda;p_i) & = (1-\lambda)\log \E e^{Z_i}\\
& \le  (1-\lambda)\log\big(1+c_0 \E(Z_i^2)\big)\\
& \le \log (1+c_*\lambda^2)\, ,
\end{align}
 whence
\begin{align*}
\phi(\lambda,u)& \ge \lambda u- \log (1+c_*\lambda^2)\, .
\end{align*}
By maximizing this expression over $\lambda$, we find that 
$P_{\ell}(\eps/2)\le \exp(-\delta_0(\eps)\ell)$ which completes the proof for the case
of i.i.d. vectors $\bY(j)$.

The case of non-identically distributed vectors follows by union bound over $a\in [K]$.
\end{proof}

\begin{lemma}\label{lemma:HittingTime-UB}
Let  $(p_i)_{i\ge 1}$, be probability distributions on 
$\cX$, with 
$\sup_{i\ge 1}\max_{x\in\cX}p_i(x)\le 1-c$, and 
$\sup_{i\ge 1}\sum_{x\in\cX} p_i(x)(\log p_i(x))^2\le C$ for constants $c, C$.

Let $(X_i)_{i\le \ell}$ be independent random variables with $X_i\sim p_i$,
$\bX= (X_1,\dots,X_{\ell})$.
 Let $\bY(j)\in\cX^{\ell}$, $j\ge 1$ be a sequence of 
i.i.d. copies of $\bX$.  
Finally, let $T:= \min\{t\ge 1: \; \bY(t) = \bX\}$.

Then, for any $\eps>0$, there exists $\delta=\delta(\eps,c,C)>0$, such that
(letting $\oH(p) := \ell^{-1}\sum_{i=1}^\ell H(p_i)$)
\begin{align}
\prob(T\ge e^{\ell [\oH(p)+\eps]})\le e^{-\delta\ell}\, .
\end{align}
\end{lemma}
\begin{proof}
The proof follows the same argument as for Lemma \ref{lemma:HittingTime}.
Denote by $\bY$ a vector distributed as $\bY(i)$. 
and define  $t_{\ell}(\eps):=e^{\ell [\oH(p)+\eps]}$,
\begin{align*}
\prob(T\ge t_{\ell}(\eps)|\bX=\bx)&=(1-\prob(\bY=\bx))^{t_\ell(\eps)}\\
& \le \exp\Big(-t_{\ell(\eps)}\prob(\bY=\bx)\Big)\, .
\end{align*}
Hence
\begin{align}
\prob\big(T\ge t_{\ell}(\eps)\big)&\le \exp\Big\{-e^{\ell\eps/2}\Big\} + 
\prob\Big(\prob(\bY=\bX|\bX)\le t_{\ell}(\eps/2)^{-1}\Big)\\
&\le   e^{-\ell\eps/2} +\tP_{\ell}(\eps/2)\, ,\\
\tP_{\ell}(u)&:=\prob\Big(\sum_{i=1}^{\ell}\log \frac{1}{p_i(X_i)}\ge  \sum_{i=1}^{\ell}H(p_i)+\ell\, u\Big)\, .
\end{align}
We  claim that, for each $u>0$, $\tP_{\ell}(u)\le e^{-\delta_0(u)\ell}$ for some $\delta_0(u)>0$.
Indeed, using again Chernoff's bound, we get, for any $\lambda\ge 0$,
$\tP_{\ell}(u) \le e^{-\ell \tphi(\lambda,u)}$, where
\begin{align}
\tphi(\lambda,u)& :=\lambda u-\frac{1}{\ell}\sum_{i=1}^{\ell}\tpsi(\lambda;p_i)\, ,\\
\tpsi(\lambda;p_i)& := \log\E \exp(\lambda W_i)\, , \;\;\;\; W_i := \log\frac{1}{p_i(X_i)}-H(p_i)\, .
\end{align}
where in the last line $X_i\sim p_i$. Under the assumptions of the lemma,
 $W_i\le   C$  almost surely and applying again Lemma \ref{lemma:QuadraticBound}, we get
$\tpsi(\lambda;p_i)\le \log(1+c_*\lambda^2)$ for $\lambda\le 1$. The proof is completed by selecting for each 
$u>0$, $\lambda>0$ so that $\lambda u-\log(1+c_*\lambda^2)>0$.
\end{proof}

We are now in position to prove Lemma \ref{lemma:NoCol}.
\begin{proof}[Proof of Lemma \ref{lemma:NoCol}]
We begin by proving the bound \eqref{eq:UB-ell}.

Fix $i\le m$, $j\le n'$, $\bu\in \cL^m$, $\delta>0$, and write $\ell = \ell(\delta,u_i)$.
Define $R_{ij}:=\{(i,j'): \max(1,j-\ell)\le j'\le j-1\}$ and 
$S_{ij}:= \{(i',j'): i'<i \mbox{ or } i'=i , j'<j-\ell\}$.
Finally, for $t\in\{0,\dots,\ell-1\}$, let
 $S_{ij}(t):= S_{ij}\cap\{(i',j'): \<i'j'\> = t \mod \ell \}$.

By union bound
\begin{align*}
\prob\big(\cE_{i,j}(\ell)\big|\bu\big)& \le A+\sum_{t=0}^{\ell-1}B(t)\, ,\\
A& :=\sum_{(rs)\in R_{ij}}\prob\big(\bX_{\<rs\>}^{\<rs\>+\ell-1} = \bX_{\<ij\>}^{\<ij\>+\ell-1}\big|\bu\big)\, ,\\
B(t)&:=\prob\Big(\exists(r,s)\in S_{ij}(t) :\;  \bX_{\<rs\>}^{\<rs\>+\ell-1} = \bX_{\<ij\>}^{\<ij\>+\ell-1}
\Big|\bu\Big)\, .
\end{align*}
Now, by the bound of Eq.~\eqref{eq:ClaimZ}, 
\begin{align}
A\le \ell\cdot(1-c)^{\ell} \le C\, N^{-\eps}\, ,
\end{align}
for suitable constants $C$, $\eps$.

Next, for any $t\in \{0,\dots,\ell-1\}$,  the vectors 
$\{\bX_{\<rs\>}^{\<rs\>+\ell-1}\}$ are mutually independent and independent
of $\{\bX_{\<ij\>}^{\<ij\>+\ell-1}\}$. Conditional on $\bu$, the coordinates of
$\bX_{\<rs\>}^{\<rs\>+\ell-1} = (X_{\<rs\>}, \cdots, X_{\<rs\>+\ell-1} )$
are independent with marginal distributions $X_{\<r's'\>}\sim Q_{x|u}(\,\cdot\,|u_{r'})$
(note that independence of the coordinates holds because $\ell<m/2$
and therefore $\bX_{\<rs\>}^{\<rs\>+\ell-1}$ does not include two entries in the same column).
Note that the collection of marginal distributions 
$Q_{x|u}(\,\cdot\,|u)$, $u\in\cL$ satisfies the conditions of Lemma 
\ref{lemma:HittingTime} by assumption. Further, the vector 
$\bX_{\<rs\>}^{\<rs\>+\ell-1}$ can have at most one of $K=|\cL|^2(\ell+1)$ distributions
(depending on the latents value and the occurrence of a line break in the block.)

Applying  Lemma 
\ref{lemma:HittingTime}, we obtain:
\begin{align}
B(t)\le e^{-\eps_0\ell}\le C\, N^{-\eps}
\end{align}
Summing over $t\in\{0,\dots,\ell-1\}$ and adjusting the constants yields the claim \eqref{eq:UB-ell}.

Next consider the bound \eqref{eq:LB-ell}. Fix $\bu\in\cL^m$, $i\le m$,$j\le n'$,
and write $\ell = \ell(-\delta,u_i)$ for brevity below
\begin{align}
\prob\big(\cE_{i,j}^c(\ell)|\bu\big)
\le \prob\Big(\forall (i',j')\in S_{ij}(t) \mbox{ s.t. }  u_{i'} = u_i, j'<n':
\bX_{\<i'j'\>}^{\<i'j'\>+\ell-1}\neq \bX_{\<ij\>}^{\<ij\>+\ell-1}\Big|\bu\Big)\, .
\label{eq:LowerBoundLemma}
\end{align}
Here $t\in\{0,\dots,\ell-1\}$ can be chosen arbitrarily. 
Let $S_{ij}(t;\bu):= \{(i',j')\in S_{ij}(t) \mbox{ s.t. }  u_{i'} = u_i, j'<n'\}$.
 Conditional 
on $\bu$, the vectors $(\bX_{\<i'j'\>}^{\<i'j'\>+\ell-1})_{(i',j')\in S_{ij}(t;\bu)}$
are i.i.d.  and independent of 
$\bX_{\<ij\>}^{\<ij\>+\ell-1}$. Further, they are distributed as
$\bX_{\<ij\>}^{\<ij\>+\ell-1}$. Finally,
\begin{align*}
N_{ij}(\bu):= \big| S_{ij}(t;\bu)
\big|
&\ge \frac{n\big(m_i(u)-C\log N\big)}{\ell}\\
&\ge \frac{m_i(u)n}{C\log N}-C'n\, .
\end{align*}
where $m_i(u)$ is the number rows $i'<i$ such that $u_{i'}=u$.
Since $i\ge i$ and $\prob(u_{i'}=u)\ge\min_{u'}\ro(u')>0$, 
by Chernoff bound there exist constants $C,c_0$ such that,
 for all $m,n$ large enough (since $i\ge m_0$)
\begin{align}
\prob\Big(N_{ij}(\bu)\ge \frac{c_0m_0n}{\log N}\Big)\ge 1-Ce^{-m_0/C}\, .
\end{align}
Further, for any $\delta>0$ we can choose positive constants $\eps_0,\eps_1>0$
such that the following holds for all $m,n$, large enough
\begin{align}
 \frac{c_0m_0n}{\log N}\ge N^{1-\eps_1}\ge e^{\ell[H(X|U=u_i)+\eps_0]} 
\end{align}
Let $T_{ij}$ be the rank of the first $(i',j')$ (moving backward)
in the set defined above such that 
$\bX_{\<i'j'\>}^{\<i'j'\>+\ell-1}=\bX_{\<ij\>}^{\<ij\>+\ell-1}$,
and $T_{ij}=\infty$ if no such vector exists.
We can continue from Eq.~\eqref{eq:LowerBoundLemma} to
get
\begin{align*}
\prob\big(\cE_{i,j}^c(\ell)\big)&\le 
\prob(T_{ij}\ge N_{ij}(\bu))\\
&\le \prob\big(T_{ij}\ge N_{ij}(\bu);\; N_{ij}(\bu) \ge e^{\ell[H(X|U=u_i)+\eps_0]} \big)+ Ce^{-m_0/C}\\
&\stackrel{(a)}{\le}  \exp\big\{-\delta_0\min_{u\in\cL}\big(\ell(-\delta;u)\big)\big\} +Ce^{-m_0/C} 
\le C N^{-\eps}\, ,
\end{align*}
where in $(a)$ we used Lemma \ref{lemma:HittingTime-UB}. This completes the proof
of Eq.~\eqref{eq:LB-ell}.
\end{proof}

\subsection{Proof of Lemma \ref{lemma:Col}}

We begin by considering the bound \eqref{eq:UB-ell-c}.
 
Fix $i\le m$, $j\le n'_c$, $\bu\in \cL^m$,
$\bv\in \cL^n$, $\delta>0$, and write $\ell = \ell_c(\delta,u_i)$,
$n'= n'_c$. By union bound:
\begin{align*}
\prob\big(\cF_{i,j}(\ell)\big|\bu,\bv\big)& = \prob\Big(
\cup_{s\in [n], |j-s|<\ell}\bB(s)\Big|\bu,\bv\Big)\, ,\\
\cB(s) &:= \Big\{\exists r<i: \;\; \bX_{\<rs\>}^{\<rs\>+\ell-1} =\bX_{\<ij\>}^{\<ij\>+\ell-1}\Big\}\, .
\end{align*}
Note that for a fixed $s$, and conditional on $\bu$, $\bv$,
the vectors $(\bX_{\<rs\>}^{\<rs\>+\ell-1})_{1\le s\le i-1}$ are mutually independent and independent of 
$\bX_{<ij>}^{<ij>+\ell-1}$. Further, $\bX_{\<rs\>}^{\<rs\>+\ell-1}$ has independent
coordinates with marginals $X_{\<r's'\>}\sim Q_{x|u}(\,\cdot\,|u_{r'},v_{s'})$
(recall that we are conditioning both on $\bu$ and $\bv$).
In particular, the marginal distributions satisfy the assumption of Lemma
 \ref{lemma:HittingTime} and the law of $\bX_{\<rs\>}^{\<rs\>+\ell-1}$ can take one of
  $K=|\cL|^2(\ell+1)$ possible values.
Letting $i-T(s)$ the last row at which  
$\bX_{\<rs\>}^{\<rs\>+\ell-1} = \bX_{\<ij\>}^{\<ij\>+\ell-1}$ (with $T(s)\ge i$ if no such row
exists), we have, for some constants $C,c_0>0$,
  \begin{align*}
  \prob\big(\cF_{i,j}(\ell)\big|\bu,\bv\big)& \le \prob\Big(
\cup_{s\in [n], |j-s|<\ell}\{T(s)\le  i-1\}\cap \{i-1\le e^{\ell [\oH-\eps]}\}\Big|\bu,\bv\Big)
+\bfone(i-1>e^{\ell [\oH-\eps]})\\
& \stackrel{(a)}{\le} 2\ell \, e^{-\ell\eps}+\bfone(m>e^{\ell [\oH-\eps]})\\
& \le Cm^{-c_0\eps}+\bfone(m>e^{\ell [\oH-\eps]})\, ,
  \end{align*}
  where in $(a)$ we used Lemma
 \ref{lemma:HittingTime}, and we defined $\oH: = \ell^{-1}\sum_{k=j}^{j+\ell-1}H(X|U=u_i,V=v_k)$.
 
 Taking expectation with respect to $\bv$, we get
  \begin{align*}
  \prob\big(\cF_{i,j}(\ell)\big|\bu\big)&\le Cm^{-c_0\eps}+
  \prob\Big(\frac{1}{\ell}\sum_{k=j}^{j+\ell-1}H(X|U=u_i,V=v_k) <
  \frac{1}{1+\delta}(H(X|U=u_i,V)+\eps)\Big)\\
  &\stackrel{(a)}{\le} Cm^{-c_0\eps}+ e^{-\ell\eps}\le  C' m^{-c_0\eps}\, ,
   \end{align*}
where in $(a)$ we used Chernoff bound. This completes the proof of Eq.~\eqref{eq:UB-ell-c}.

Finally, the proof Eq.~\eqref{eq:LB-ell-c} is similar to the one of  Eq.~\eqref{eq:LB-ell}.
We fix $\bu\in\cL^m$, $i\le m$, $j\le n'_c$,
and write $\ell = \ell_c(-\delta,u_i)$.
\begin{align}
\prob\big(\cF_{i,j}^c(\ell)|\bu\big)
\le \prob\Big(\forall i'<i u_{i'} = u_i:
\bX_{\<i'j\>}^{\<i'j\>+\ell-1}\neq \bX_{\<ij\>}^{\<ij\>+\ell-1}\Big|\bu\Big)\, .
\label{eq:LowerBoundLemma-c}
\end{align}
Let $S^c_{ij}(\bu):= \{(i',j)\; \mbox{ s.t. }  u_{i'} = u_i, i'<i\}$.
 Conditional 
on $\bu,\bv$, the vectors $(\bX_{\<i'j'\>}^{\<i'j'\>+\ell-1})_{(i',j')\in S^c_{ij}(u)}$
are i.i.d.  and independent copies of 
$\bX_{\<ij\>}^{\<ij\>+\ell-1}$. Finally,
$N^c_{i}(\bu):= \big| S_{ij}^c(\bu) \big|$
is the number rows $i'<i$ such that $u_{i'}=u$.
By Chernoff bound there exist constants $C,c_0$ such that,
 for all $m,n$ large enough (recalling that we need only to consider $i\ge m_0$)
\begin{align}
\prob\Big(N^c_{i}(\bu)\ge c_0m_0\Big)\ge 1-Ce^{-m_0/C}\, .
\end{align}
Since $m_0 \ge m^{1-o_n(1)}$, for any $\delta>0$ we can choose constants
$\eps_0,\eps_1>0$ so that
\begin{align}
c_0m_0 \ge m^{1-\eps_1} \ge e^{\ell[H(X|U=u_i,V)+2\eps_0]}\, .
\end{align}
Recall the definition $\oH: = \ell^{-1}\sum_{k=j}^{j+\ell-1}H(X|U=u_i,V=v_k)$.
By an an application of Chernoff bound

\begin{align}
\prob\Big(N^c_{i}(\bu)\ge e^{\ell[H(X|U=u_i,V)+\eps_0]}\Big)\ge 1-Cm^{-\eps}-Ce^{-m_0/C}\, .
\end{align}
Let $T_{i}$ be the rank of the first $i'$ (moving backward)
in the set defined above such that 
$\bX_{\<i'j\>}^{\<i'j\>+\ell-1}=\bX_{\<ij\>}^{\<ij\>+\ell-1}$,
and $T_{i}=\infty$ if no such vector exists.
From Eq.~\eqref{eq:LowerBoundLemma-c} we get
\begin{align*}
\prob\big(\cF_{i,j}^c(\ell)\big)&\le 
\prob(T_{i}\ge N_{i}(\bu))\\
&\le \prob\big(T_{i}\ge N_{i}(\bu);\; N_{i}(\bu) \ge e^{\ell[H(X|U=u_i,V)+\eps_0]} \big)+ 
Cm^{-\eps}\\
&\stackrel{(a)}{\le}  
\exp\big\{-\delta_0\min_{u\in\cL}\big(\ell_c(-\delta;u)\big)\big\} +C m^{-\eps}   
\le 2C m^{-\eps}\, ,
\end{align*}
where in $(a)$ we used Lemma \ref{lemma:HittingTime-UB}. 

\section{Proofs for latent-based encoders}
\label{sec:ProofLatent}

\subsection{Proof of Lemma \ref{lemma:Consistency}}
\label{sec:ProofConsistency}

\subsubsection{General bound \eqref{eq:GeneralUB-Consistency}}
\label{sec:ProofGeneralUBConsistency}

From Eq.~\eqref{eq:LatentGen}, we get
\begin{align}
\Rate_{\slat}(\bX) = 
\frac{1}{mn\log_2|\cX|}\Big\{ \len({\sf header} ) +
\len(\zip_{\cL}(\hbu))+\len(\zip_{\cL}(\hbv))
+\sum_{u,v\in\cL}\len(\zip_{\cX}(\hbX(u,v)))\Big\} \wedge 1\, ,
\end{align}
where $\hbX(u,v) := \vec\big(X_{ij}:\; \hu_i(\bX)=u,\hv_j(\bX)=v\big)$
are the estimated blocks of $\bX$. 
Note that this rate depends on the base compressors $\zip_{\cL}$,
$\zip_{\cX}$ but we will omit these from our notations.

Define the `ideal' expected compression rate (i.e. the rate achieved by a compressor 
that is given the latents):
\begin{align*}
\Rate_{\#} & :=\frac{1}{mn\log_2|\cX|}\Big\{
\E[\len({\sf header} )] +
\E[\len(\zip(\bu))]+\E[\len(\zip(\bv))]
+\sum_{u,v\in\cL}\E[\len(\zip(\bX(u,v)))]
\Big\}\, .
\end{align*}
Since $\Rate_{\slat}(\bX)\le 1$ by construction, we have
\begin{align*}
\E\,\Rate_{\slat}(\bX) &\le 
\E\,
\big\{\Rate_{\slat}(\bX) \, 
\bfone_{\Acc_U(\bX;\hbu)=1}\bfone_{\Acc_V(\bX;\hbv)=1}\big\}
+\prob\big( \Acc_U(\bX;\hbu)<1\big) +
\prob\big( \Acc_V(\bX;\hbv)<1\big)\\
&\stackrel{(*)}{\le} \Rate_{\#}+\prob\big( \Acc_U(\bX;\hbu)<1\big) +
\prob\big( \Acc_V(\bX;\hbv)<1\big)\, ,
\end{align*}
where in step $(*)$ we bounded
$\E[\len(\zip(\hbu))\bfone_{\Acc_U(\bX;\hbu)=1}]=
\E[\len(\zip(\bu))\bfone_{\Acc_U(\bX;\hbu)=1}]\le\E[\len(\zip(\bu))]$,
because, on the event $\{\Acc_U(\bX;\hbu)=1\}$, $\hbu$ coincides with
$\bu$ up to relabelings, and the compressed length is invariant under relabelings.
Similar arguments were applied to $\len(\zip(\bv))$ and $\len(\zip(\bX(u,v)))$.

We now have, by the definition of $\Delta_{\zip}(N;k)$ in Eq.~\eqref{eq:DefOverhead}, 
\begin{align}
\frac{\E[\len(\zip(\bu))]}{mn\log_2|\cX|}&\le \frac{H(U)}{n\log_2|\cX|}+
+\frac{1}{n}\Delta_{\zip}(m\wedge n;\{r,c\})\, ,\\
\frac{\E[\len(\zip(\bv))]}{mn\log_2|\cX|}&\le \frac{H(V)}{m\log_2|\cX|}+
+\frac{1}{m}\Delta_{\zip}(m\wedge n;\{r,c\})\, ,\\
\frac{\E[\len(\zip(\bX(u,v)))|\bu,\bv]}{mn\log_2|\cX|}
&\le \hro(u)\hco(v)\frac{H(X|U=u,V=v)}{\log_2|\cX|}
+\Delta_{\zip}(c\cdot mn;\{Q(\,\cdot\,|u,v)\}_{i,v\in\cL})\, ,
\end{align}
where in the last line $\hr$ is the empirical distribution of the row latents and 
$\hc$ is the empirical distribution of the column latents.
By taking expectation in the last expression, we get
\begin{align}
\sum_{u,v\in\cL}\frac{\E[\len(\zip(\bX(u,v)))|\bu,\bv]}{mn\log_2|\cX|}
\le \frac{H(X|U,V)}{\log_2|\cX|}
+|\cL|^2\Delta_{\zip}(c\cdot mn;\{Q(\,\cdot\,|u,v)\}_{u,v\in\cL})\, .
\end{align}
Finally, the header contains $|\cL|^2+2$ integers of maximum size $mn$,
whence $\len({\sf header} )\le 4\log_2(mn)$. We conclude that
\begin{align*}
\Rate_{\#}\le &\frac{1}{\log_2|\cX|}\Big\{H(X|U,V)+\frac{1}{n}H(U)+\frac{1}{n}H(V)\Big\}
 +\frac{2\log_2(mn)}{mn}\\
&+
 |\cL|^2\Delta_{\zip}(c\cdot mn;\{Q(\,\cdot\,|u,v)\}_{u,v\in\cL})+2\Delta_{\zip}(m\wedge n;\{r,c\})
 \, .
\end{align*}
The claim \eqref{eq:GeneralUB-Consistency} follows from the first bound in
Eq.~\eqref{eq:Fano} noticing that, under the stated assumptions on $m,n$,
\begin{align}
\frac{1}{n}
 \big[\bentro(\eps_U)+\eps_u\log(|\cL|-1)\big]\le \eps_U\le \prob\big( \Acc_U(\bX^{m,n};\hbu)<1\big)\, .
 \end{align}

\subsubsection{Redundancy bounds for  specific encoders: 
Eqs.~\eqref{eq:LZ-overhead}--\eqref{eq:ANS-overhead}}
\label{sec:Redundancy}

\noindent{\bf LZ coding.} 
 Let $\bX^N = (X_1,\dots,X_N)$ be a vector with i.i.d. symbols $X_i\sim q$ 
 with $q$ a probability distribution over $\cX$. 
 The analysis is similar to the one in Appendix \ref{sec:ProofLZ}, and we will adopt the same 
 notations here. There are two important differences:
 data are i.i.d. (not matrix-structured) and  we want to derive a sharper estimate 
 (not just the entropy term, but bounding the overhead as well).
 
We define $L_k(\bX^N)$, $T_k(\bX^N)$ as per Eqs.~\eqref{eq:LK-def}, \eqref{eq:TK-def}. 
We let $(k(1),\dots, k(M_N))$ be the values taken by $k$ 
in the while loop of the Lempel-Ziv pseudocode of Section \ref{sec:LZ}. In particular 
\begin{align}
k(1) & = 1\, ,\\
k(\ell+1) & = k(\ell)+L_{k(\ell)}(\bX^N)\, ,\\
k(M_N)  & = N\, .
\end{align}
(We set $k(0) = 0$ by convention.) 
Therefore the total length of the code is 
\begin{align*}
\len(\LZ(\bX^{N})) & = M_N\lceil\log_2 (N+|\cX|)\rceil +\sum_{\ell=1}^{M_N} \len(\elias(L_{k(\ell)}))\\
&\le M_N\lceil\log_2 (N+|\cX|)\rceil + 2 \sum_{\ell=1}^{M_N} \log_2(L_{k(\ell)})\\
& \le  M_N\lceil\log_2 (N+|\cX|)\rceil + 2 M_N \log_2(N/M_N)\, ,
\end{align*}
where the last step follows by Jensen's inequality. By one more application of Jensen,
we obtain
\begin{align}
\E\, \Rate_{\LZ}(\bX^{N})\le
\frac{1}{\log_2|\cX|} \cdot \frac{\E M_N}{N}\cdot \big\{\lceil\log_2 (N+|\cX|)\rceil
 + 2 \log_2(N/\E M_N)\big\}\, .\label{RateLZBound}
\end{align}

Define the set of break points and bad positions as 
\begin{align}
S_N &:=\big\{k(1),k(2),\dots,k(M_N)\big\}\, ,\\
B_N(\ell)& := \big\{k\in [N/2,N]:\; L_k(\bX^N)<\ell\big\}\, .
\end{align}
Note that $S_N =S^{\le }_N\cup S^{> }_N$ where:
\begin{align}
S^{\le}_N &:=\Big\{k(j):\; j\le M_N, k(j-1)\le \lfloor N/2\rfloor \Big\}\, ,\;\;\;\;\;
S^{>}_N &:=\Big\{k(j):\; j\le M_N, k(j-1)> \lfloor N/2\rfloor\Big\}\, .
\end{align}
Further $|S^{\le}_N|\ed M_{\lfloor N/2\rfloor}$ and, for any $\ell\in\naturals$,
\begin{align}
\frac{N}{2} \ge \sum_{k\in S_N^{>}} L_{k} \ge \big(|S^>_N|-|B_N(\ell)|\big)\ell\, .
\end{align}
Therefore, 
\begin{align}
\E\, |S^{>}_N|&\le \frac{N}{2\ell}+ \E\, |B_N(\ell)|\nonumber\\
& \le \frac{N}{2\ell}+\sum_{k = \lceil N/2\rceil}^N \prob\big( L_k(\bX^N)< \ell\big)\, .
\label{eq:ESN}
\end{align}
We claim that this implies, for $C_0=20 c_*\log|\cX|$  and $\log N \ge (2\log(2/H(q)))^2$,
\begin{align}
\frac{1}{N}\E\, |S^{>}_N|&\le \frac{ H(q)}{2\log_2 N}+C_0
 \frac{(\log\log_2 N)^{1/2}}{(\log_2 N)^{3/2}}=: \psi(\log_2 N)\, .\label{eq:Technical-LZ-Claim}
\end{align}
Before proving this claim, let us show that it implies the thesis. 
Recall that $M_N=|S_N|$ and $S_N =S^{\le }_N\cup S^{> }_N$ where
 $|S^{\le}_N|\ed M_{\lfloor N/2\rfloor}$. Therefore, we have proven
 \begin{align}
 \E M_N &\le N\psi(\log_2 N) +  \E M_{\lfloor N/2\rfloor} \nonumber\\
&  \le \sum_{\ell=0}^{K-1}N_\ell\psi(\log_2 N_\ell) + \E M_{N_K}\, , \label{eq:RecMN}
 \end{align}
where we defined recursively $N_{0}=N$, $N_{\ell+1}=\lfloor N/2\rfloor$, and
$K:= \min\{\ell: \; \log_2 N_{\ell}< (2\log(2/H(q)))^2 \}$. Of course, $M_{N_K}\le N_K\le
\exp((2\log 2/H(q))^2)$. Further $\uN_{\ell}\le N_{\ell} \le \oN_{\ell}$, 
where $\uN_{0}=\oN_0=N$ and $\oN_{\ell+1}=\oN_{\ell}/2$, $\uN_{\ell+1}=
(\uN_{\ell}-1)/2$ for  $\ell\ge 0$. We thus get  $\uN_{\ell}= (N+1)2^{-\ell}-1$,
$\oN_{\ell} = N\, 2^{-\ell}$ and therefore 
 \begin{align*}
 \frac{1}{N}\sum_{\ell=0}^{K-1}N_\ell\psi(\log_2 N_\ell)  &\le  \frac{1}{N}\sum_{\ell=0}^{\infty}\oN_{\ell}
 \psi(\log_2 \uN_{\ell})\\
 &\le \frac{H(q)}{2}\sum_{\ell=0}^{\infty}2^{-\ell}\frac{1}{\log_2 (N 2^{-\ell}-1)} +
 C_0\sum_{\ell=0}^{\infty}2^{-\ell}\frac{(\log\log_2 N)^{1/2}}{(\log_2 (N 2^{-\ell}-1))^{3/2}}\\
 &\le \frac{H(q)}{\log_2 N} + 2 C_0 \frac{(\log\log_2 N)^{1/2}}{(\log_2 N )^{3/2}}\, .
 \end{align*}
 Substituting in Eq.~\eqref{eq:RecMN}, we get
 \begin{align*}
\frac{1}{N} \E M_N &\le  \frac{H(q)}{\log_2 N} + 2 C_0 \frac{(\log\log_2 N)^{1/2}}{(\log_2 N )^{3/2}}
+\frac{1}{N}\exp\big\{\big(2\log (2/H(q))\big)^2\big\}\\
 &\le \frac{H(q)}{\log_2 N} + 3 C_0 \frac{(\log\log_2 N)^{1/2}}{(\log_2 N )^{3/2}}\, ,
\end{align*}
where the last inequality follows for $N\ge \exp\big\{\big(4\log (2/H(q))\big)^2\big\}$
(noting that $C_0>1$). Finally, the desired bound \eqref{eq:LZ-overhead}
follows by substituting the last estimate in Eq.~\eqref{RateLZBound}.

We are left with the task of proving claim \eqref{eq:Technical-LZ-Claim}.
Fix any $k$, $\lceil N/2\rceil \le k\le N$ and write $q^{\ell}$ for
the product distribution $q\times \cdots\times q$ ($\ell$ times). 
Setting $H=H_{\snats}(q)$ (measuring here entropy in nats), for any $\delta>0$,
\begin{align}
 \prob\big( L_k(\bX^N)< \ell\big) &=\prob\Big(\bX_{i}^{i+\ell-1}\neq \bX_{k}^{k+\ell-1}
 \; \forall i<k\Big)\\
 &\le   
 \sum_{\bx^{\ell}\in\cX^{\ell}}\prob\big( Z_{k}(\bx^{\ell}) = 0\big)\cdot
 \prob\big( \bX_{k}^{k+\ell-1} = \bx^{\ell}\big)\nonumber\\
 & \le \sum_{\bx^{\ell}\in\cX^{\ell}}q^{\ell}(\bx^{\ell})\big(1-q^{\ell}(\bx^{\ell})\big)^{N/2\ell}
 \nonumber\\
 &\le \sum_{\bx^{\ell}\in\cX^{\ell}}q^{\ell}(\bx^{\ell})\; 
 \bfone\big(q^{\ell}(\bx^{\ell})\le e^{-\ell[H+\delta]}\big)
 + \exp\Big\{-\frac{N}{2\ell}\cdot e^{-\ell[H+\delta]}\Big\}\nonumber\\
 &=:P_{\le}(\ell;\delta)+P_{>}(\ell,N;\delta)\label{eq:Psum}
 \,.
 \end{align}
 By Chernoff bound
\begin{align*}
 P_{\le}(\ell;\delta) &\le e^{-\ell\max_{\lambda>0}\psi_{\delta}(\lambda)}\, ,\label{eq:Chernoff}\\
 \psi_{\delta}(\lambda) & := \lambda[H+\delta]-\log\Big\{\sum_{x\in\cX}q(x)^{1-\lambda}\Big\}\, .
 \end{align*}
 Note that $\lambda \mapsto\psi_{\delta}(\lambda)$ is continuous, concave,
 with $\psi_{\delta}'(0)=\delta$, $\psi_{\delta}(0)=0$,  $\psi_{\delta}(1)=\delta+H-\log|\cX|$. 
 Hence (assuming $H<\log|\cX|$ because otherwise there is nothing to prove) for all $\delta$
  small enough $\psi$ is maximized for $\lambda\in (0,1)$. 
 Further, defining the random variable $Q=q(x)$ for
 $x\sim\Unif(\cX)$, 
 \begin{align}
 \psi''_{\delta}(\lambda) &= -\frac{\E[Q^{1-\lambda} (\log Q)]}{\E[Q^{1-\lambda}]}
 +\Big[\frac{\E[Q^{1-\lambda} (\log Q)^2]}{\E[Q^{1-\lambda}]}\Big]^2 \\
 &\ge  -\frac{\E[Q^{1-\lambda} (\log Q)^2]}{\E[Q^{1-\lambda}]}\\
 &\ge  -\E[(\log Q)^2] =: -\underline{c}_*(q)\, .
 \end{align}
 Here the last inequality holds because $Q\mapsto Q^{1-\lambda}$ is monotone increasing
 (for $\lambda\in [0,1]$) and $Q\mapsto (\log Q)^2$ is monotone decreasing over $Q\in [0,1]$,
 and therefore $\E[Q^{1-\lambda} (\log Q)^2]\le \E[Q^{1-\lambda}]\,\E[(\log Q)^2]$.
 In what follows, we set $c_* :=  \underline{c}_*(q)\wedge 1$.

 Hence $\psi_{\delta}(\lambda) \ge \delta\lambda-c_*\lambda^2/2$ for $\lambda\in [0,1]$ and
 therefore using Eq.~\eqref{eq:Chernoff},
\begin{align*}
 P_{\le}(\ell;\delta) &\le \exp\Big\{-\ell \min\Big(\frac{\delta^2}{2c_*};\delta-\frac{c_*}{2}\Big)\Big\}
 \, .
 \end{align*}
 
 Substituting in Eq.~\eqref{eq:Psum}, and using this in Eq.~\eqref{eq:ESN},
 we get, for $\delta\in [0,c_*]$:
\begin{align*}
 \frac{1}{N}\E|S^{>}_N|&\le \frac{1}{2\ell}+ \exp\Big\{-\frac{\ell\delta^2}{2c_*}\Big\}+
 \exp\Big\{-\frac{N}{2\ell}\cdot e^{-\ell[H+\delta]}\Big\}
 \end{align*}
We set 
\begin{align*}
\ell = \frac{\log N}{H}(1-\eps)\, ,\;\;\;\;\; \delta = \frac{1}{2} H\eps\, ,
\end{align*}
for $\eps\le (2c_*/H)\wedge (1/2)$. Substituting in the previous bound, we get
\begin{align*}
 \frac{1}{N}\E|S^{>}_N|&\le \frac{H}{2\log N}(1+2\eps)+
 \exp\Big\{-\frac{H\eps^2\log N}{16c_*}\Big\}+
 \exp\Big\{-\frac{HN^{(\eps+\eps^2)/2} }{2\log N}\Big\}\, .
\end{align*}
We finally select $\eps = c_0(c_*\log|\cX|/H)(\log\log N/\log N)^{1/2}$, with $c_0$ a
sufficiently small absolute constant. Substituting above, 
\begin{align*}
 \frac{1}{N}\E|S^{>}_N|- \frac{H}{2\log N} &\le c_0\, c_*\, \log|\cX|
 \frac{(\log\log N)^{1/2}}{(\log N)^{3/2}}+
 \exp\Big\{-\frac{c_0^2(\log|\cX|)^2c_*}{16 H}\log\log N\Big\}\\
 &\phantom{AAAAAA}+ \exp\Big\{-\frac{H}{2\log N}e^{(c_0c_*\log|\cX|/2H)\sqrt{\log N}}\Big\}\\
 &\le c_0\, c_*\, \log|\cX|
 \frac{(\log\log N)^{1/2}}{(\log N)^{3/2}} +(\log N)^{-c_0^2/16}
 +\exp\Big\{-\frac{H}{2\log N}e^{(c_0/2)\sqrt{\log N}}\Big\}\, .
\end{align*}
Setting $c_0=5$, we get
\begin{align*}
 \frac{1}{N}\E|S^{>}_N|- \frac{H}{2\log N} &\le 6\, c_*\, \log|\cX|
 \frac{(\log\log N)^{1/2}}{(\log N)^{3/2}}
  +\exp\Big\{-\frac{H}{2\log N}e^{2\sqrt{\log N}}\Big\}\, ,
\end{align*}
whence, the claim \eqref{eq:Technical-LZ-Claim} follows for $\log N\ge (\log(2/H))^2$. 

\noindent{\bf Arithmetic Coding.} 
In Arithmetic Coding (AC) we encode the empirical distribution of $\bX^N$
$\hq_N(x): =  N^{-1}\sum_{i\le N}\bfone_{X_i=x}$, and then encode $\bX^N$
in at most $-\log_2\hq^{N}(\bX^N)+1$ bits.
The  encoding of $\hq_N$ amounts to encoding the $|\cX|-1$ integers $N\hq_N(x)$,
$x\in\cX\setminus\{0\}$ (assuming that $0\in\cX$, one of the counts can be obtained by difference). 
We thus have
\begin{align*}
\len(\AC(\bX^{N})) & \le  -\log_2\hq_N^{\otimes N}(\bX^N) +1 +
\sum_{x\in\cX}\len(\elias(N\hq_N(x)))\\
& \le  -\log_2\hq^{\otimes N}_N(\bX^N) + 2|\cX|\log_2 N\\
& = \sum_{i=1}^N-\log_2\hq_N(X_i) + 2|\cX|\log_2 N\\
& = N\, H(\hq_N)  + 2|\cX|\log_2 N\, .
\end{align*}
Taking expectations
\begin{align*}
\E\Rate_{\AC}(\bX^{N}) & \le \frac{\E \, H(\hq_N)}{\log_2 |\cX|} + \frac{2|\cX|\log_2 N}{N\log_2|\cX|}\\
& \le \frac{H(q)}{\log_2 |\cX|} + \frac{2|\cX|\log_2 N}{N\log_2|\cX|}\, .
\end{align*}

\noindent{\bf ANS Coding.}  The bound \eqref{eq:ANS-overhead} follows for range ANS coding 
from the analysis of \cite{duda2009asymmetric,duda2013asymmetric,kosolobov2022efficiency},
where encoding of empirical distributions are analyzed as for AC coding. 

\subsection{Proof of Theorem \ref{thm:MainLatent}}

The proof consists in applying Lemma \ref{lemma:Consistency}
and showing that $\prob\Big( \Acc_U(\bX^{m,n};\hbu)>0\Big)\le\log(mn)/mn$,
$\prob\Big( \Acc_V(\bX^{m,n};\hbv)>0\Big)\le\log(mn)/mn$.

In what follows we will assume without loss of generality $m\le n$, and 
recall that $|\cL|=k$,  identifying $\cL= \{1,\dots, k\}$. We will assume
$k$ fixed. We will use $C, c, c',\dots$ for constants that might depend on $k$, 
as well as the constant $c_0$ in the statement
in ways that we do ot track.

 We will show that these bounds hold conditional on $\bu$, $\bv$,
on the events $\min_u \hro(u), \min_v\hro(v)\ge c/2$ which holds with probability
at least $1-\exp(-c'm)\ge 1- \log(mn)/mn$. Hence, hereafter we will treat $\bu,\bv$ as deterministic.
Recall that $\bM\in \reals^{m\times n}$ is the matrix entries
with 
\begin{align*}
M_{ij} = 
\psi(X_{i,j})\, ,
\end{align*}
and let $\bM_* = \E\{\bM\}$. We collect a few facts about $\bM$ and its expectation.

\vspace{0.15cm}

\noindent{\bf Singular values.} Note that $\bM_*$ takes the form 
\begin{align}
\bM_* = \bL\bPsi\bR^{\sT}\, ,
\end{align} 
where $\bPsi\in \reals^{r\times r}$ is a matrix with entries $\bPsi_{u,v} = \opsi(u,v)$,
$\bL\in\{0,1\}^{m\times r}$, with $L_{ij} = 1 \Leftrightarrow u_i=j$,
and $\bR\in\{0,1\}^{n\times r}$, with $R_{ij} = 1 \Leftrightarrow v_i=j$.
Define $\bL = \bL_0\bD_L^{1/2}$ where $\bD_L$ is a diagonal matrix with 
$(\bD_{L})_{ii} = m\hro(i)$, and analogously  $\bR = \bR_0\bD_L^{1/2}$,
and introduce the singular value decomposition $\bD_L^{1/2} \bPsi\bD_R^{1/2}=
\obA\bSigma\obB^{\sT}$. We then have the singular value decomposition
\begin{align}
\bM_* = \bA_*\bSigma\bB_*^{\sT}\, , \;\;\;\; \bA_* = \bL_0\obA\,,\;\;  \bB_* = \bR_0\obB\, .
\end{align} 
Therefore $\sigma_{k}(\bM_*) \ge
 \sigma_{\min}(\bD_L)^{1/2}\sigma_{\min}(\bPsi) \sigma_{\min}(\bD_R)^{1/2}$
 (here and below $\sigma_{k}$ denotes the $k$-th largest singular value)
 and using the assumptions on $\hro,\hco$, 
\begin{align}
\sigma_{k}(\bM_*)& \ge c\mu\sqrt{mn}\, .
\end{align}

\vspace{0.15cm}

\noindent{\bf Concentration.} $\bM-\bM_*$ is a centered matrix with independent entries with variance 
bounded by $\sigma^2$ and entries bounded by $1$ (by the assumption $|\psi(x)|\le 1$).
By matrix Bernstein inequality there exists a universal constant $C$
such that the following holds with probability at least $1-(100n)^{-2}$:
\begin{align}
\|\bM-\bM_*\|_{\op}\le C\max\Big(\sigma\sqrt{n\log n}; \log n\Big)\, .
\end{align}

\vspace{0.15cm}

\noindent{\bf Incoherence.} Since all the entries of $\bM_*$ are bounded by $1$, we get
\begin{align}
\|\bM_*\|_{2\to\infty}\vee \|\bM^{\sT}_*\|_{2\to\infty} \le \nu\sqrt{n}\, .
\end{align}

\vspace{0.15cm}

\noindent{\bf Row concentration.} For any $i\le m$, and any $\bW\in\reals^{n\times k}$ fixed,
 with probability at least $1-(100n)^{-5}$:
\begin{align*} 
\|(\bM-bM_*)_{i,\cdot}\bW\|_2\le C \max\Big(\sigma\|\bW\|_F\sqrt{\log n} ; 
\|\bW\|_{2\to\infty}\log n\Big)\, .
\end{align*}
Defining $\Delta_*:=\sigma_{k}(\bM_*) \ge c\mu\sqrt{mn}$, this implies
\begin{align*} 
\|(\bM-bM_*)_{i,\cdot}\bW\|_2&\le \Delta_*\|\bW\|_{2\to\infty}\varphi\left(\frac{\|\bW\|_{F}}{\sqrt{n}\|\bW\|_{2\to\infty}}\right)\, ,\\
\varphi(x) &:= \frac{C}{\mu\sqrt{mn}}\max\big( x\sigma\sqrt{n\log n} ;\log n \big)\, .
\end{align*}

\vspace{0.15cm}

Given these, we apply \cite[Corollary 2.1]{abbe2020entrywise}, with the following
estimates of various parameters (see \cite{abbe2020entrywise} for definitions):
\begin{align*}
\Delta_* & \asymp \mu\sqrt{mn}\, ,\\
\gamma & \lesssim \max\left(
\frac{\sigma}{\mu}\sqrt{\frac{\log n}{m}} ;
\frac{\log n}{\mu\sqrt{mn}} \right)\lesssim \frac{\sigma}{\mu}\sqrt{\frac{\log n}{m}}\, ,\\
\|\bM_*\|_{2\to\infty}\vee \|\bM^{\sT}_*\|_{2\to\infty} & \le \nu\sqrt{n} \lesssim \gamma\Delta_*\, ,\\
\varphi(1) &\lesssim \frac{\sigma}{\mu}\sqrt{\frac{\log n}{m}} \, ,\\
\varphi(\gamma) &\lesssim  \max\left(\frac{\sigma^2}{\mu^2}\frac{\log n}{m}; 
\frac{\log n}{\mu\sqrt{mn}}\right) \, ,\\
\kappa & \asymp 1\, .
\end{align*}
Then  \cite[Corollary 2.1]{abbe2020entrywise} implies that there exists a $k\times k$
orthogonal matrix $\tbQ$ such that
\begin{align}
\frac{\|\bA- \bA_*\tbQ\|_{2\to \infty}}{\|\bA_*\|_{2\to \infty}}
&\lesssim (1+\varphi(1))(\gamma+\varphi(\gamma))+\varphi(1)\\
&\lesssim\frac{\sigma}{\mu}\sqrt{\frac{\log m}{n}}\, .
\end{align}
Recall that $\bA_* = \bL_0\obA$ and $\obA$ is an othogonal matrix. 
Therefore, there exists an orthogonal matrix $\bQ$ such that (with the desired probability):
\begin{align}
\frac{\|\bA- \bL_0\bQ\|_{2\to \infty}}{\|\bL_0\|_{2\to \infty}}
\lesssim\frac{\sigma}{\mu}\sqrt{\frac{\log m}{n}}\, .
\end{align}
Further, the $i$-th row of $\bL_0$ is 
\begin{align}
(\bL_0)_{i,\cdot} = \frac{1}{\sqrt{m\hro(u_i)}}\bfe^{\sT}_{u_i}=: z(u_i)\bfe^{\sT}_{u_i}\, .
\end{align}
Hence, for any $i$,  $\sqrt{c_0}\|\bL_0\|_{2\to \infty}\le \|(\bL_0)_{i,\cdot}\|\le
 \|\bL_0\|_{2\to \infty}$. Denoting by $\bq_{j}$ the $j$-th row of $\bQ$, we thus get
 for all $i$, 
\begin{align}
\frac{\|\ba_i- z(u_i)\bq_{u_i}\|_{2}}{\|\ba_i\|_{2}}
\lesssim\frac{\sigma}{\mu}\sqrt{\frac{\log m}{n}}\, .
\end{align}
The claim follows immediately using the fact that $\sqrt{c_0}\max_uz(u)\le \min_u z(u)\le \max_u z(u)$.
%
%

\end{document}

%% file: commands.tex
\newcommand{\bea}{\begin{eqnarray}}
\newcommand{\eea}{\end{eqnarray}}
\newcommand{\<}{\langle}
\renewcommand{\>}{\rangle}
\newcommand{\E}{{\mathbb E}}

\def\Acc{\, {\sf Err}}
\def\entro{{\sf h}}
\def\enc{{\sf Enc}}
\def\zip{{\sf Z}}
\def\row{{\mbox{\tiny\rm row}}}
\def\col{{\mbox{\tiny\rm col}}}
\def\slat{{\mbox{\tiny\rm lat}}}
\def\snats{{\mbox{\tiny\rm nats}}}

\def\uN{\underline{N}}
\def\oN{\overline{N}}

\def\bR{{\boldsymbol R}}
\def\bU{{\boldsymbol U}}
\def\bV{{\boldsymbol V}}
\def\bX{{\boldsymbol X}}
\def\hbX{\hat{\boldsymbol X}}

\def\hu{\hat{u}}
\def\hv{\hat{v}}

\def\opsi{\overline{\psi}}
\def\bPsi{{\boldsymbol \Psi}}

\def\bL{{\boldsymbol L}}

\def\sTV{\mbox{\tiny \rm TV}}
\def\sinit{\mbox{\tiny \rm init}}

\def\Unif{{\sf Unif}}

\def\eps{{\varepsilon}}

\def\rE{{\rm E}}

\def\cuP{\mathscrsfs{P}}

\def\cuQ{\mathscrsfs{Q}}

\def\plain{{\sf plain}}

\def\LZ{{\sf LZ}}
\def\ANS{{\sf ANS}}
\def\AC{{\sf AC}}

\def\concat{{\oplus}}

\def\oH{\overline{H}}
\def\avp{\overline{p}}

\def\tP{\tilde{P}}
\def\tpsi{\tilde{\psi}}
\def\tphi{\tilde{\phi}}

\def\btau{{\boldsymbol{\tau}}}
\def\bsigma{{\boldsymbol{\sigma}}}

\def\bxi{{\boldsymbol{\xi}}}
\def\bfe{{\boldsymbol{e}}}

\def\bZ{{\boldsymbol{Z}}}
\def\bSigma{{\boldsymbol{\Sigma}}}

\def\bQ{{\boldsymbol{Q}}}
\def\bM{{\boldsymbol{M}}}

\def\bV{{\boldsymbol{V}}}

\def\hq{{\widehat{q}}}
\def\hQ{{\widehat{Q}}}

\def\hbu{{\boldsymbol{\widehat{u}}}}
\def\hbv{{\boldsymbol{\widehat{v}}}}

\def\cF{{\mathcal F}}
\def\cG{{\mathcal G}}
\def\cT{{\mathcal T}}

\def\cX{{\mathcal X}}

\def\SBM{\mbox{\tiny\rm SBM}}

\def\op{\mbox{\tiny\rm op}}

\def\init{\mbox{\tiny\rm init}}

\def\naturals{{\mathbb N}}
\def\reals{{\mathbb R}}

\def\sT{{\sf T}}

\def\bq{{\boldsymbol{q}}}
\def\bv{{\boldsymbol{v}}}
\def\bz{{\boldsymbol{z}}}
\def\bx{{\boldsymbol{x}}}
\def\ba{{\boldsymbol{a}}}
\def\bb{{\boldsymbol{b}}}
\def\tbb{\tilde{\boldsymbol{b}}}
\def\bA{\boldsymbol{A}}
\def\bB{\boldsymbol{B}}
\def\obA{\boldsymbol{\overline{A}}}
\def\obB{\boldsymbol{\overline{B}}}
\def\tbQ{\boldsymbol{\tilde{Q}}}

\def\tbX{\tilde{\boldsymbol{X}}}
\def\tbY{\tilde{\boldsymbol{Y}}}
\def\bX{\boldsymbol{X}}
\def\bY{\boldsymbol{Y}}
\def\bW{\boldsymbol{W}}
\def\prob{{\mathbb P}}
\def\E{{\mathbb E}}

\def\<{\langle}
\def\>{\rangle}

\def\ed{\stackrel{{\rm d}}{=}}

\def\cL{{\cal L}}

\def\bw{{\boldsymbol{w}}}

\def\cE{{\mathcal E}}

\def\ro{q_{\mbox{\small\rm r}}}
\def\co{q_{\mbox{\small\rm c}}}
\def\hro{\hat{q}_{\mbox{\small\rm r}}}
\def\hco{\hat{q}_{\mbox{\small\rm c}}}

\def\Rate{{\sf R}}
\def\DRR{{\sf DRR}}
\def\hr{\hat{r}}
\def\hc{\hat{c}}
\def\hp{\hat{p}}
\def\hhp{\overline{p}}

\def\bD{{\boldsymbol{D}}}

\def\bu{{\boldsymbol{u}}}
\def\b0{{\boldsymbol{0}}}

\def\tba{{\boldsymbol{\tilde{a}}}}
\def\tbb{{\boldsymbol{\tilde{b}}}}

\def\Var{{\rm Var}}

\def\sinit{\mbox{\tiny\rm init}}

\def\bfone{{\boldsymbol 1}}

\def\rP{{\rm P}}

\def\cA{{\mathcal A}}
\def\cZ{{\mathcal Z}}

\def\cB{{\mathcal B}}

\def\ess{{\rm ess}}

\def\conc{\oplus}
\def\elias{{\sf elias}}
\def\len{{\sf len}}

\def\bentro{{\sf h}}
\def\Perm{{\mathfrak S}}

\def\vec{{\sf vec}}